\documentclass[10pt,aps,pra,twocolumn,nofootinbib,superscriptaddress,floatfix]{revtex4-2}

\usepackage{fix-cm}

\usepackage{amsmath, amssymb}
\usepackage{xcolor}
\usepackage{titletoc}
\usepackage{url}
\usepackage{dirtytalk}
\usepackage{graphicx} 
\usepackage{subcaption}
\usepackage{wrapfig} 
\usepackage[T1]{fontenc}
\usepackage{tcolorbox}
\usepackage{mathtools}
\usepackage{svg}
\usepackage{bm}
\usepackage{multirow}
\usepackage{algorithm}
\usepackage{algorithmic}
\usepackage{xfrac}
\usepackage{qcircuit}
\usepackage{braket}
\usepackage{hyperref}
\hypersetup{breaklinks=true}
\usepackage{cleveref}
\allowdisplaybreaks
\hypersetup{
    colorlinks=true,
    linkcolor=blue,
    citecolor=red,
    urlcolor=magenta
}
\DeclareMathOperator{\Tr}{Tr}
\newtheorem{theorem}{Theorem}

\newtheorem{lemma}[theorem]{Lemma}

\newtheorem{proposition}[theorem]{Proposition}
\newtheorem{remark}[theorem]{Remark}

\newenvironment{proof}[1][Proof]{\noindent\textbf{#1.} }{\ \rule{0.5em}{0.5em}}
\begin{document}

\title{Quantum natural gradient with thermal-state initialization}
%Quantum natural gradient for thermal-state initialization of parameterized quantum circuits
%Mixed-state quantum natural gradient
%Quantum geometry-aware optimization of mixed-state parameterized quantum circuits
\author{Michele Minervini}
\affiliation{\textit{School of Electrical and Computer Engineering, Cornell University, Ithaca, New York 14850, USA}}
\author{Dhrumil Patel}  
\affiliation{Department of Computer Science, Cornell University, Ithaca, New York 14850, USA}
\author{Mark M. Wilde}
\affiliation{\textit{School of Electrical and Computer Engineering, Cornell University, Ithaca, New York 14850, USA}}

\date{\today}

\begin{abstract}
Parameterized quantum circuits (PQCs) are central to variational quantum algorithms (VQAs), yet their performance is hindered by complex loss landscapes that make their trainability challenging. Quantum natural gradient descent, which leverages the geometry of the parameterized space through quantum generalizations of the Fisher information matrix, offers a promising solution but has been largely limited to pure-state scenarios, with only approximate methods available for mixed-state settings. This paper addresses this question, originally posed in [Stokes \textit{et al}., \textit{Quantum} \textbf{4}, 269 (2020)], by providing exact methods to compute three quantum generalizations of the Fisher information matrix—the Fisher–Bures, Wigner–Yanase, and Kubo–Mori information matrices—for PQCs initialized with thermal states. We prove that these matrix elements can be estimated using quantum algorithms combining the Hadamard test, classical random sampling, and Hamiltonian simulation. By broadening the set of quantum generalizations of Fisher information and realizing their unbiased estimation, our results enable the implementation of quantum natural gradient descent algorithms for mixed-state PQCs, thereby enhancing the flexibility of optimization when using VQAs. Another immediate consequence of our findings is to establish fundamental limitations on the ability to estimate the parameters of a state generated by an unknown PQC, when given sample access to such a state.%, through the quantum Cramer--Rao bound.
\end{abstract}

\maketitle

\tableofcontents

\allowdisplaybreaks

\section{Introduction}

\subsection{Motivation}

Parameterized quantum circuits (PQCs) have emerged as a key area of research in quantum computing due to their flexibility and ease of implementation~\cite{Mitarai2018quantum_circuit_learning,Cerezo2021vqa}. These circuits are foundational for variational quantum algorithms (VQAs), enabling a hybrid computational approach where quantum circuits are iteratively optimized using classical methods to tackle a wide range of problems~\cite{Lavrijsen2020opt_pqc,Sweke2020stoch_grad_desc,Peruzzo2014vqe,McClean2016theore_vqa}. However, the practical implementation of PQCs faces several challenges. Key issues include poor local minima~\cite{Bittel2021vqaNPhard,Anschuetz2022qva_traps}, limited 
expressivity~\cite{tikku2022circuitdepthversusenergy, Holmes2022expressibility_barrenplateaus}, and the widespread occurrence of barren plateaus~\cite{McClean2018barren_plateaus,larocca2024review_barren_plateaus}. These obstacles hinder the efficient training and scalability of VQAs, raising important questions about the long-term viability of PQCs for solving large-scale problems~\cite{cerezo2022challengesQML,zimborás2025myths_quantum_computation_fault}. Addressing these limitations is critical for advancing the field and realizing the potential of quantum computing in practical applications.

One of the most prominent challenges in VQAs is the highly complex loss landscape, which is characterized by flat regions (barren plateaus) and numerous local minima~\cite{You2021many_local_min,Anschuetz2022qva_traps}. This structure makes gradient-based optimization inefficient, leading to slow convergence and increased computational costs~\cite{Bittel2021vqaNPhard,Xiaozhen2022opt_landscape}. The difficulty stems from the strong dependence of gradient evolution on parameterization, which complicates the search for an optimal trajectory from initialization to convergence. Thus, avoiding barren plateaus and local minima traps remains a significant hurdle~\cite{McArdle2019,Wierichs2020nat_grad_avoid_local_min}, highlighting the need for innovative strategies to improve the training and performance of VQAs.

A potential approach to mitigating these issues is natural gradient descent—a metric-aware optimization method that incorporates the geometry of the state space, allowing for more judicious navigation of the loss landscape~\cite{amari1998nat_grad}. Originally introduced in classical machine learning, natural gradient descent replaces the standard Euclidean metric with the Fisher information me\-tric, which captures the geometry of the parameterized space~\cite{amari1998nat_grad}. By ensuring that updates in the parameter space follow the steepest descent direction as defined by the state-space geometry, this method has demonstrated significant practical advantages in classical settings. For example, it has been shown to achieve faster convergence and improved optimization performance in deep learning~\cite{pascanu2014revisitingnaturalgradientdeep}, reinforcement learning~\cite{kakade2002natgrad_policygrad,Peters2005nat_grad_reinforcement_learn}, and training of large-scale neural networks~\cite{Martens2020}. 

These successes have motivated the adaptation of natural gradient descent to the quantum setting~\cite{Stokes2020qng,sbahi2022provablyefficientvariationalgenerative,Koczor2022qng_noisy,sohail2024quantumnaturalstochasticpairwise,patel2024naturalgradientparameterestimation,minervini2025eqbm}, where multiple quantum generalizations of the Fisher information matrix exist, derived from different smooth divergences of quantum states (see also \cite{Liu2019q_Fisher_multi_par_est,Sidhu2020,Jarzyna2020,Meyer2021fisher_info_NISQ,sbahi2022provablyefficientvariationalgenerative,scandi2024quantumfisherinformationdynamical} for various reviews of quantum generalizations of Fisher information). In this context, numerical works have shown that quantum natural gradient outperforms standard optimizers for certain problems~\cite{yamamoto2019naturalgradientvariationalquantum,Stokes2020qng,Wierichs2020nat_grad_avoid_local_min, Koczor2022qng_noisy,meyer2023qnpg}, as it follows optimization paths that align with the quantum state space geometry. Furthermore, a convergence analysis under specific conditions has recently been developed~\cite{sohail2024quantumnaturalstochasticpairwise}.

Despite its advantages, a key challenge in applying quantum natural gradient descent lies in the computational cost of estimating quantum generalizations of the Fisher information matrix for general PQCs~\cite{Meyer2021fisher_info_NISQ,Cerezo2021sub_opt_Fisher,Gacon2021stoch_approx_Fisher,Beckey2022vqa_est_Fisher_info,Straaten2021cost_nat_grad}. Since the matrix elements must be evaluated at each optimization step, this process can be expensive. However, the relative overhead of this computation often becomes negligible as the optimization progresses, particularly when compared to the increasing cost of gradient estimation near convergence~\cite{Straaten2021cost_nat_grad}. Furthermore, various approaches have been introduced to reduce the overhead of quantum natural gradient descent \cite{Gacon2021stoch_approx_Fisher,Fitzek2024optimizing,sohail2024quantumnaturalstochasticpairwise,halla2025estimationquantumfisherinformation,halla2025modifiedconjugatequantumnatural}.

Beyond computational challenges, existing implementations of quantum natural gradient descent for PQCs have been limited to pure states, where the Fubini--Study metric---a measure of distinguishability of pure quantum states based on the geometry of the quantum state space---can be computed and used to determine parameter updates. Extending this approach to mixed states has remained an open problem~\cite{Stokes2020qng}, as it requires the de\-ve\-lop\-ment of new methods to characterize and compute the quantum Fisher information matrix in mixed-state sce\-na\-rios. Some methods have been proposed to approximate the quantum Fisher information matrix in such settings~\cite{Gacon2021stoch_approx_Fisher,Koczor2022qng_noisy,Beckey2022vqa_est_Fisher_info}, relying mainly on the Fisher–Bures metric. However, exact results have been lacking, and the applicability of natural gradient descent in mixed-state scenarios has remained unexplored. 

This limitation is particularly significant because many optimization tasks in quantum computing require or benefit from mixed-state solutions~\cite{Biamonte2017qml,brandao2017quantumspeedupssemidefiniteprogramming,Dallaire_Demers_2018}. For instance, in variational solvers for semi-definite programs (SDPs)~\cite{patel2024vqa_sdp,chen2023qslackslackvariableapproachvariational,liu2025quantumthermodynamicssemidefiniteoptimization}, the optimization variable is itself a density operator, and typical constraints push solutions toward mixed states. Mixed states also arise naturally in quantum machine learning, where they improve robustness to noisy data and act as effective regularizers in supervised tasks~\cite{LaRose2020data_encoding_qml, beer2020qnn}. In such cases, natural gradient descent tailored to mixed states could enable more efficient and accurate optimization by accounting for the geometry of the mixed-state space, which is inherently non-Euclidean. Therefore, developing new methods to characterize and compute quantum generalizations of the Fisher information matrix for mixed-state PQCs would extend the applicability of quantum natural gradient descent. By leveraging its ability to navigate the cost landscape more effectively than standard optimizers, this approach could enhance the robustness of variational quantum algorithms in real-world applications.

\subsection{Main results}

\setlength{\tabcolsep}{5pt} % Default value: 6pt
\renewcommand{\arraystretch}{1.8}
\begin{table*}
\centering
\begin{tabular}
[c]{|l|l|c|c|}\hline\hline
Quantity & Formula & Theorem & Quantum Circuit\\\hline\hline 
Fisher--Bures & $I^{\operatorname{FB}}_{ij}(\phi) =\left\langle \left[  \left[  \mathcal{U}^\dag_{R_j}\!(H_j) , G \right],\Phi\!\left(\mathcal{U}^\dag_{R_i}\!(H_i)\right) \right]  \right\rangle _{\rho}$ & Theorem~\ref{thm:FB} & Figure~\ref{fig:FB} \\\hline
\multirow{2}{*}{Wigner--Yanase} & $I^{\operatorname{WY}}_{ij}(\phi) =  4  \left\langle \left\{ \mathcal{U}^\dag_{R_i}\!(H_i) , \mathcal{U}^\dag_{R_j}\!(H_j) \right\}\right\rangle_\rho$ 
& \multirow{2}{*}{Theorem~\ref{thm:WY}} & Figure~\ref{fig:WY-1} \\
& $\hspace{1.45cm} - \, 8 \Tr\!\left[  \mathcal{U}^\dag_{R_i}\!(H_i) \sqrt{\rho} \ \mathcal{U}^\dag_{R_j}\!(H_j) \sqrt{\rho} \right] $ & & Figure~\ref{fig:WY-2}\\\hline
Kubo--Mori & $I^{\operatorname{KM}}_{ij}(\phi) = \left\langle \left[ \left[ \mathcal{U}^\dag_{R_j}\!(H_j) , G \right] , \mathcal{U}^\dag_{R_i}\!(H_i) \right] \right\rangle_\rho$ & Theorem~\ref{thm:KM} & Figure~\ref{fig:KM}
\\\hline\hline
\end{tabular}
\caption{Summary of our analytical results for the matrix elements of the Fisher--Bures, Wigner--Yanase, and Kubo--Mori information matrices for parameterized quantum circuits initialized in a thermal state.}
\label{table:FB-WY-KM-results}
\end{table*}

In this paper, we present analytical formulas for three different quantum generalizations of the Fisher information matrix---namely, the Fisher--Bures (Theorem~\ref{thm:FB}), Wigner--Yanase (Theorem~\ref{thm:WY}), and Kubo--Mori (Theorem~\ref{thm:KM}) information matrices---for PQCs initialized with a thermal state. Along with these findings, we also construct quantum algorithms for estimating these information matrices. These results extend beyond existing approaches, which are typically limited to pure states or approximate methods for the mixed-state case~\cite{Stokes2020qng,Meyer2021fisher_info_NISQ,Koczor2022qng_noisy}. We prove that these matrix elements can be estimated  via a combination of the Hadamard test~\cite{Cleve1998}, classical random sampling, and Hamiltonian simulation~\cite{lloyd1996universal,childs2018toward}. Our analytical findings are summarized in Table~\ref{table:FB-WY-KM-results}, where we also highlight their similarity to related expressions for quantum evolution machines, a special case of evolved quantum Boltzmann machines, as presented in \cite[Table~1]{minervini2025eqbm}. 

The estimation of these quantum generalizations of the Fisher information matrix enables three different implementations of the quantum natural gradient descent algorithm for PQCs, allowing each implementation to execute the natural gradient update step on a quantum computer. Additionally, these results establish fundamental limitations on the ability to estimate the parameters of a PQC, when given access to quantum states generated by the circuit, through the multiparameter quantum Cramer--Rao bound.

While our quantum algorithms enable unbiased estimation of the quantum Fisher information matrices, their efficient implementation relies on delicate control of the system. In particular, estimating the Fisher--Bures and Kubo--Mori information matrix elements requires access to the Hamiltonian that prepares the initial thermal state, while the Wigner--Yanase information matrix elements necessitate the ability to prepare a canonical purification of the initial thermal state. We note that there has been significant progress in thermal-state preparation quantum algorithms in recent years~\cite{chen2023q_Gibbs_sampl,chen2023thermalstatepreparation,rajakumar2024gibbssampling,bergamaschi2024gibbs_sampling,chen2024sim_Lindblad,rouzé2024efficientthermalization,bakshi2024hightemperaturegibbsstates,ding2024preparationlowtemperaturegibbs}. 

Despite these requirements, our results mark a significant theoretical advancement, as they provide the first exact expressions and quantum algorithms for estimating quantum generalizations of the Fisher information matrix for PQCs starting from a general mixed quantum state. Furthermore, our results expand the set of available quantum Fisher information matrices beyond the commonly used Fubini–Study metric~\cite{Stokes2020qng,Meyer2021fisher_info_NISQ,Koczor2022qng_noisy}, offering greater flexibility in choosing the most suitable metric for a given optimization problem (see Section~\ref{sec:natural_gradient} for further discussion). Moreover, as shown in \cite[Corollary 8]{minervini2025eqbm}, the Fisher–Bures and Wigner–Yanase information matrices differ by no more than a factor of two in the matrix (Loewner) order, making them essentially interchangeable for natural gradient descent. This allows one to choose the simpler-to-estimate information matrix without compromising the overall optimization trajectory.

\section{Problem setting}

In this section, we establish the definitions and pre\-li\-mi\-na\-ry results that underpin our main results. A key aspect of variational quantum algorithms is the choice of the quantum circuit ansatz, which determines how quantum states are parameterized. The ansatz selection plays a critical role in both the expressivity of the quantum states explored by the algorithm and the optimization behavior (trainability)~\cite{Sim2019ansatz_pqc,Du2020expressivity,Nakaji2021expressibility}.
Various ansatz designs have been proposed for different applications~\cite{farhi2014qaoa,Kandala2017hardware_efficient_ansatz,Mitarai2018quantum_circuit_learning,Grimsley2019,Hadfield2019qaoa,Wiersema2020ham_var_ansatz}, all of which fundamentally involve combinations of rotational gates and, in some cases, fixed gates such as CNOT or SWAP. Here we consider a general layered parameterized ansatz with $N$ qubits and $J$ layers, defined as
\begin{equation}\label{eq:param_circ}
    U(\phi) \coloneqq U_J(\phi_J)\cdot\cdot\cdot U_2(\phi_2) U_1(\phi_1),
\end{equation}
where $\phi \coloneqq (\phi_1, \dots, \phi_J)^{\mathsf{T}}$ is a vector of independent parameters. Each unitary $U_j(\phi_j)$ is composed of a parameterized part generated by a Hermitian operator $H_j$ and an unparameterized unitary $V_j$: 
\begin{equation}
    U_j(\phi_j) \coloneqq e^{-i \phi_j H_j} V_j,\label{eq:U_j_def}
\end{equation}
where the gate generators $\left\{ H_j \right\}_j$ are not necessarily commuting.

Here, we assume that the initial state $\rho$ is a thermal state of the form
\begin{equation}\label{eq:thermal_state}
    \rho \coloneqq \frac{e^{-G}}{Z},
\end{equation}
where $Z \coloneqq \Tr[e^{-G}]$ is the partition function and $G$ is the underlying Hamiltonian, with the temperature factor absorbed into $G$ for simplicity. To estimate the quantum information matrices of PQCs in Section~\ref{sec:infor_matrices}, we further assume that $G$ is accessible alongside $\rho$, meaning we have access to the Hamiltonian underlying the initial state in~\eqref{eq:thermal_state}. While this assumption is non-trivial, it holds for $k$-local or sparse Hamiltonians, where $G$ can be learned efficiently from measurements of the Gibbs state in certain regimes~\cite{Haah2024learning_Ham,gu2024ham_learn,bakshi2024learning_ham}. For example, high-temperature Gibbs states permit Hamiltonian recovery via local marginals~\cite{Haah2024learning_Ham,gu2024ham_learn}, while fault-tolerant methods (e.g., block-encoding~\cite{Gilyen2019q_sing_value_tranf,Martyn2021}  or quantum phase estimation~\cite{bakshi2024learning_ham}) offer theoretical pathways for broader settings.

The initial state in~\eqref{eq:thermal_state} then evolves under the parameterized quantum circuit as
\begin{equation}\label{eq:state}
    \rho(\phi) \coloneqq U(\phi) \rho U^\dagger(\phi).
\end{equation}

Now we introduce some notation to simplify the fol\-lo\-wing results. We define the unitary components occurring after and before the $j$-th gate as
\begin{align}
    U_{L_{j}} &\coloneqq U_J \cdot\cdot\cdot U_{j}, \label{eq:partial_gates_left}\\
    U_{R_j} & \coloneqq U_{j}\cdot\cdot\cdot U_1, \label{eq:partial_gates_right}
\end{align}
where we have omitted the explicit dependence on $\phi$ in $U_{L_{j}}$ and $U_{R_j}$ for simplicity.
Using this notation, we define two quantum channels that play a central role in the subsequent analysis: 
\begin{align}
    \mathcal{U}_{L_j}\!(X) & \coloneqq U_{L_j} X U^\dagger_{L_j},\label{eq:q_channel_left}\\
    \mathcal{U}^\dag_{R_j}\!(X) & \coloneqq U^\dagger_{R_j} X U_{R_j}, \label{eq:q_channel_right}
\end{align}
which describe the conjugation of a generic operator $X$ by the unitary operators $U_{L_{j}}$ and $U_{R_j}$, respectively.

In this work, we focus on the optimization of PQCs; thus, the derivatives with respect to the variational parameters $\{\phi_j\}_j$ are crucial~\cite{Izmaylov2021param_shift_rule_general}. Due to the structure of the parameterized circuit in~\eqref{eq:param_circ}, the gradient of this unitary satisfies
\begin{equation}
    \frac{\partial}{\partial\phi_{j}} U(\phi) = -i U_{L_{j+1}} H_j U_{R_j},
\end{equation}
where we used the notation in~\eqref{eq:partial_gates_left}--\eqref{eq:partial_gates_right}.
Consequently, the gradient of the parameterized state in~\eqref{eq:state} is given by the following lemma:
\begin{lemma} \label{lem:part_der}
    The partial derivative of the parameterized state $\rho(\phi)$, defined in~\eqref{eq:state}, with respect to the variational parameter $\phi_j$ is given by
    \begin{align}\label{eq:grad_pqc}
        \frac{\partial}{\partial\phi_{j}} \rho(\phi) = i \left[ \rho(\phi) , \mathcal{U}_{L_{j+1}}\!(H_j) \right],
    \end{align}
    where $\mathcal{U}_{L_{j+1}}\!$ is the quantum channel defined in~\eqref{eq:q_channel_left}. The full gradient is then the vector whose components are the partial derivatives: $\nabla_\phi \rho(\phi) = \left( \frac{\partial}{\partial\phi_{1}} \rho(\phi), \dots, \frac{\partial}{\partial\phi_{J}} \rho(\phi) \right)$.
\end{lemma}
\begin{proof}
    The result follows from the structure of the parameterized circuit in~\eqref{eq:param_circ} and the chain rule. For a detailed derivation, see Appendix~\ref{app:grad_pqc}.
\end{proof}
\medskip

Before presenting the main contributions of this paper—namely, analytical expressions and quantum circuits for estimating the matrix elements of the Fisher--Bures (Section~\ref{sec:FB}), Wigner--Yanase (Section~\ref{sec:WY}), and Kubo--Mori (Section~\ref{sec:KM}) information matrices for PQCs (summarized in Table~\ref{table:FB-WY-KM-results})—we establish some additional notation and preliminary results that serve as the foundation for the following developments. We express an eigendecomposition of the initial state  $\rho$ as
\begin{equation}
    \rho = \sum_k \lambda_k |\tilde{k}\rangle\!\langle \tilde{k}|,
\end{equation} 
which, recalling that it is defined as a thermal state in~\eqref{eq:thermal_state}, also takes the form
\begin{equation}
    \rho = \frac{1}{Z}\sum_k e^{-\mu_k} |\tilde{k}\rangle\!\langle \tilde{k}|, 
\end{equation}
where we have used a spectral decomposition of the Hamiltonian $\smash{G=\sum_k \mu_k |\tilde{k}\rangle\!\langle \tilde{k}|}$.
An eigendecomposition of the parameterized state $\rho(\phi)$, defined in~\eqref{eq:state}, is then given by 
\begin{equation}\label{eq:spcetr_decomp}
    \rho (\phi)=\sum_{k}\lambda_{k}|k\rangle\!\langle k|,
\end{equation} 
where $|k\rangle = U(\phi) |\tilde{k}\rangle$. For simplicity, we omit the explicit dependence of $\ket{k}$ on the parameter vector $\phi$. 
Finally, we present a result that will be instrumental for evaluating the information matrix elements of PQCs: 
\begin{lemma}\label{lem:matr_part_der}
    Using the shorthand $\partial_{j}\equiv\frac{\partial}{\partial\phi_{j}}$, the matrix elements of the gradient of the parameterized state $\rho(\phi)$ defined in~\eqref{eq:state} are given by
    \begin{align}
    \langle k| \partial_{j}\rho(\phi)|\ell\rangle & = i \left( \lambda_k - \lambda_\ell \right) \langle k| \mathcal{U}_{L_{j+1}} \!(H_j) |\ell\rangle\\
    & = i \left( \lambda_k - \lambda_\ell \right) \langle \tilde{k}| \mathcal{U}^\dag_{R_j}\!(H_j)  |\tilde{\ell}\rangle, \label{eq:last_aux_info}
    \end{align}
    where $\mathcal{U}_{L_{j+1}}\!$ and $\mathcal{U}^\dag_{R_j}\!$ are the quantum channels defined in~\eqref{eq:q_channel_left} and~\eqref{eq:q_channel_right}, respectively.
\end{lemma}
\begin{proof}
    From Lemma~\ref{lem:part_der}, consider the following:
    \begin{align}
        \langle & k| \partial_{j}\rho(\phi)|\ell\rangle \notag\\
        & =  \langle k| i \left[ \rho(\phi) , \mathcal{U}_{L_{j+1}} \right] |\ell\rangle \\
        & = i \langle k|  \rho(\phi) \mathcal{U}_{L_{j+1}} |\ell\rangle  - i \langle k| \mathcal{U}_{L_{j+1}} \rho(\phi)|\ell\rangle\\
        & = i \left( \lambda_k - \lambda_\ell \right) \langle k| \mathcal{U}_{L_{j+1}} \!(H_j) |\ell\rangle\\
        & = i \left( \lambda_k - \lambda_\ell \right) \langle \tilde{k}| U^\dagger(\phi) U_{L_{j+1}} H_j U^\dagger_{L_{j+1}} U(\phi)|\tilde{\ell}\rangle\\
        & = i \left( \lambda_k - \lambda_\ell \right) \langle \tilde{k}| U^\dagger_{R_{j}} H_j U_{R_{j}} |\tilde{\ell}\rangle\\
        & = i \left( \lambda_k - \lambda_\ell \right) \langle \tilde{k}| \mathcal{U}^\dag_{R_j}\!(H_j)  |\tilde{\ell}\rangle.
    \end{align}    
    This completes the proof.
\end{proof}
\medskip

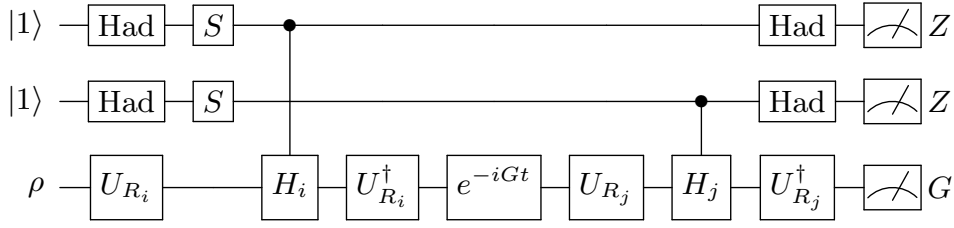
\begin{figure*}
    \centering
    \centering
    \scalebox{1.29}{
    \Qcircuit @C=0.9em @R=1.0em {
        \lstick{\ket{1}} & \gate{\operatorname{Had}} & \gate{S} & \ctrl{2} & \qw & \qw & \qw & \qw & \gate{\operatorname{Had}} & \meter &\rstick{\hspace{-1.2em}Z} \\
        \lstick{\ket{1}} & \gate{\operatorname{Had}} & \gate{S} & \qw & \qw & \qw & \qw & \ctrl{1} & \gate{\operatorname{Had}} & \meter & \rstick{\hspace{-1.2em}Z}\\
        \lstick{\rho} & \gate{\vphantom{U^\dagger_{R_j}}U_{R_i}} & \qw & \gate{\vphantom{U^\dagger_{R_j}}H_{i}} & \gate{\vphantom{U^\dagger_{R_j}}U^\dagger_{R_i}} & \gate{\vphantom{U^\dagger_{R_j}}e^{-iGt}} & \gate{\vphantom{U^\dagger_{R_j}}U_{R_j}} & \gate{\vphantom{U^\dagger_{R_j}}H_{j}} & \gate{\vphantom{U^\dagger_{R_j}}U^\dagger_{R_j}} & \meter & \rstick{\hspace{-1.2em}G}
    }
    }
    \caption{Quantum circuit that realizes an unbiased estimate of $\frac{1}{4} \left\langle \left[  \left[  \mathcal{U}^\dag_{R_j}\!(H_j) , G \right],\Phi\!\left(\mathcal{U}^\dag_{R_i}\!(H_i)\right) \right]  \right\rangle _{\rho}$, which corresponds to the Fisher--Bures information matrix elements up to a factor of $\frac{1}{4}$. For each run of the circuit, the time $t$ is sampled independently at random from the probability density $p(t)$ in~\eqref{eq:high-peak-tent-density}. For details of the algorithm, see Appendix~\ref{app:FB_algo}.}
    \label{fig:FB}
\end{figure*}

\section{Information matrices for PQCs}\label{sec:infor_matrices}
A quantum generalization of Fisher information is a measure of how much a parametrized quantum state changes under a change of a parameter. As such, it is closely related to smooth divergences, which quantify the di\-stin\-gui\-sha\-bi\-li\-ty of two states. Specifically, for a parameterized, differentiable family of states $\left( \sigma(\gamma) \right)_{\gamma \in \mathbb{R}^L}$, where $\gamma$ is an \mbox{$L$-dimensional} parameter vector with $L\in \mathbb{N}$, the \mbox{$\boldsymbol{D}$-based} Fisher
information matrix is defined as
\begin{equation}
\left[  I^{\boldsymbol{D}}(\gamma)\right]  _{ij}\coloneqq \left.
\frac{\partial^{2}}{\partial\varepsilon_{i}\partial\varepsilon_{j}
}\boldsymbol{D}(\sigma(\gamma)\Vert\sigma(\gamma+\varepsilon))\right\vert
_{\varepsilon=0},\label{eq:D-based-Fisher-matrix}
\end{equation}
where $\boldsymbol{D}$ is a smooth divergence (see \cite[Section~III-A]{minervini2025eqbm} for the definition of smooth divergence that we are using here). If we would like to denote the dependence of the matrix $I^{\boldsymbol{D}}(\gamma)$ on the family $\left(  \sigma(\gamma)\right)_{\gamma\in\mathbb{R}^{L}}$, then we employ the notation
\begin{equation}
I^{\boldsymbol{D}}\!\left(  \gamma;\left(  \sigma(\gamma)\right)  _{\gamma
\in\mathbb{R}^{L}}\right)  \coloneqq I^{\boldsymbol{D}}(\gamma).
\end{equation} 
Reference~\cite[Section~III]{minervini2025eqbm} provides a detailed analysis of the connection between a smooth divergence $\boldsymbol{D}$ and its associated $\boldsymbol{D}$-based Fisher information matrix, along with their properties. For further background on quantum generalizations of Fisher information, see also~\cite{Bengtsson2006,Liu2019q_Fisher_multi_par_est,Sidhu2020,Jarzyna2020,Meyer2021fisher_info_NISQ,sbahi2022provablyefficientvariationalgenerative,scandi2024quantumfisherinformationdynamical}.

In this paper, we focus on three specific divergences: the quantum relative entropy, twice the negative logarithm of the Holevo fidelity~\cite{Kholevo1972}, and twice the negative logarithm of the Uhlmann fidelity~\cite{Uhl76}. These divergences respectively lead to the Kubo–Mori, Wigner–Yanase, and Fisher–Bures information matrices (see~\cite[Definition~6]{minervini2025eqbm} for formal definitions). In what follows, we adopt their explicit expressions, derived in~\cite[Theorem~7]{minervini2025eqbm}, for a general parameterized, differentiable family of states $\left( \sigma(\gamma) \right)_{\gamma \in \mathbb{R}^L}$ (see also~\cite{Sidhu2020,sbahi2022provablyefficientvariationalgenerative}).

\subsection{Fisher--Bures information matrix of PQCs with thermal-state initialization}\label{sec:FB}

The main result of this section, Theorem~\ref{thm:FB}, provides an explicit analytical expression for the Fisher--Bures information matrix elements of PQCs initialized with a thermal state. Additionally, we present a quantum circuit designed to estimate the terms in this expression (see Figure~\ref{fig:FB}). Before stating the theorem, let us recall that the Fisher–Bures information matrix elements are given by~\cite[Theorem~7]{minervini2025eqbm}
\begin{equation}
    I^{\operatorname{FB}}_{ij}(\phi) =\sum_{k,\ell}\frac{2}{\lambda_{k}+\lambda_{\ell}}\langle k|\partial_{i}\rho(\phi)|\ell\rangle\langle\ell|\partial_{j}\rho(\phi)|k\rangle,\label{eq:FB-def}
\end{equation}
where $\rho(\phi)$ is expressed in terms of its spectral decomposition, as defined in~\eqref{eq:spcetr_decomp}.

The following theorem makes use of the quantum channel $\Phi$, defined as
\begin{equation}\label{eq:Phi}
        \Phi(X) \coloneqq \int_\mathbb{R} dt \ p(t) \ e^{-iGt} X e^{iGt},
\end{equation}
where
    \begin{equation}\label{eq:high-peak-tent-density}
        p(t) \coloneqq \frac{2}{\pi} \ln \left | \coth\!{\left( \frac{\pi t}{2} \right)}\right |
    \end{equation}
    is a probability density function known as the \textit{high-peak-tent} density~\cite{patel2024quantumboltzmannmachine}.

\begin{theorem}\label{thm:FB}
    Let $\rho$ be an initial thermal state as defined in~\eqref{eq:thermal_state}. For a general layered parameterized quantum circuit of the form in~\eqref{eq:param_circ}, the Fisher--Bures information matrix elements are given by:
    \begin{align}
    \label{eq:FB}
        I^{\operatorname{FB}}_{ij}(\phi) =\left\langle \left[  \left[  \mathcal{U}^\dag_{R_j}\!(H_j) , G \right],\Phi\!\left(\mathcal{U}^\dag_{R_i}\!(H_i)\right) \right]  \right\rangle _{\rho},
    \end{align}
    where $\hspace{0.03cm}\mathcal{U}^\dagger_{R_j}\!$ denotes the unitary channel defined in~\eqref{eq:q_channel_right} and the quantum channel $\Phi$ is defined in~\eqref{eq:Phi}.
\end{theorem}

\begin{proof}
    See Appendix~\ref{app:FB-proof}.
\end{proof}
\medskip

\begin{figure*}
    \centering
    \begin{subfigure}{\textwidth}
        \centering
        \scalebox{1.5}{ % Adjust scale
        \Qcircuit @C=1.0em @R=1.0em {
            \lstick{\ket{0}} & \gate{\operatorname{Had}} & \ctrl{1} & \qw & \gate{\operatorname{Had}} & \meter  & \rstick{\hspace{-1.2em}Z} \\
            \lstick{\rho} & \gate{\vphantom{U^\dagger_{R_j}}U_{R_j}} & \gate{\vphantom{U^\dagger_{R_j}}H_j} & \gate{U^\dagger_{R_j}} & \gate{\vphantom{U^\dagger_{R_j}}U_{R_i}} & \meter & \rstick{\hspace{-1.2em}H_{i}}
        }
        }
        \vspace{5pt}
        \caption{Quantum circuit that realizes an unbiased estimate of $\frac{1}{2}\left\langle \left\{ \mathcal{U}^\dag_{R_i}\!(H_i) , \mathcal{U}^\dag_{R_j}\!(H_j) \right\}\right\rangle_\rho$, which corresponds to the first term of the Wigner--Yanase information matrix elements up to a factor of $\frac{1}{8}$. For details of the algorithm, see Appendix~\ref{app:WY-first_term}.}
        \label{fig:WY-1}
    \end{subfigure}
    \vspace{5pt}

    \begin{subfigure}{\textwidth}
        \centering
        \scalebox{1.5}{ % Adjust scale to fit the column width
        \Qcircuit @C=1.0em @R=1.0em {
            \lstick{} &  \gate{U_{R_i}} & \meter & \rstick{\hspace{-1.2em}H_i} \\
            \lstick{} & \gate{\overline{U}_{R_j}} &   \meter & \rstick{\hspace{-1.2em}\overline{H}_{j}}
            \inputgroupv{1}{2}{0.9em}{1.4em}{\hspace{-0.2em}\ket{\psi}}
        }
        }
        \vspace{5pt}
        \caption{Quantum circuit that realizes an unbiased estimate of $\Tr\!\left[  \mathcal{U}^\dag_{R_i}\!(H_i) \sqrt{\rho} \  \mathcal{U}^\dag_{R_j}\!(H_j) \sqrt{\rho} \right]$, which corresponds to the second term of the Wigner--Yanase information matrix elements up to a factor of $-\frac{1}{8}$. For details of the algorithm, see Appendix~\ref{app:WY-second_term}.}
        \label{fig:WY-2}
    \end{subfigure}
    \caption{Quantum circuits involved in the estimation of $I^{\operatorname{WY}}_{ij}(\phi)$.}
    \label{fig:WY}
\end{figure*}
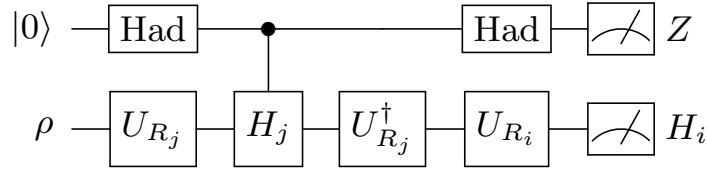
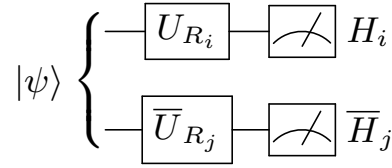

The expression in~\eqref{eq:FB} closely resembles the Fisher--Bures information matrix elements derived in \cite[Theorem~11]{minervini2025eqbm} for quantum evolution machines, a special case of evolved quantum Boltzmann machines. The only distinction lies in the form of the derivative of PQCs, given in~\eqref{eq:grad_pqc}, which differs from that of quantum evolution machines. As a result, the quantum channel $\Psi^\dagger_\phi$ ap\-pea\-ring in the Fisher--Bures information matrix elements of quantum evolution machines is replaced by $\mathcal{U}^\dag_{R_j}$, defined in~\eqref{eq:q_channel_right}. This similarity extends to the Wigner--Yanase (Theorem~\eqref{thm:WY}) and Kubo--Mori (Theorem~\eqref{thm:KM}) information matrices as well. In Appendix~\ref{app:QEM_PQC}, we provide a detailed discussion of the relationship between quantum evolution machines and PQCs.

The quantity in~\eqref{eq:FB} can be estimated on a quantum computer using the circuit shown in Figure~\ref{fig:FB}, with $S$ denoting the phase gate
\begin{equation}
    S \coloneqq
    \begingroup
    \renewcommand{\arraystretch}{1.2} % Adjust row spacing
    \setlength{\arraycolsep}{6pt}     % Adjust column spacing
    \begin{pmatrix}
    1 & 0 \\
    0 & i
    \end{pmatrix}.
    \endgroup
    \end{equation}
A detailed description of the quantum algorithm is presented in Appendix~\ref{app:FB_algo}.

\subsection{Wigner--Yanase information matrix of PQCs with thermal-state initialization}\label{sec:WY}

Theorem~\ref{thm:WY}, the central result of this section, provides an explicit analytical expression for the Wigner–Yanase information matrix elements of PQCs initialized with a thermal state. Furthermore, we introduce a quantum circuit that plays a role in estimating these matrix elements. Before presenting the theorem, let us recall that the Wigner–Yanase information matrix elements are given by~\cite[Theorem~7]{minervini2025eqbm}
\begin{equation}
    I^{\operatorname{WY}}_{ij}(\phi) =\sum_{k,\ell}\frac{4}{\left(  \sqrt{\lambda_{k}}+\sqrt{\lambda_{\ell}}\right)  ^{2}}\langle k|\partial_{i}\rho(\phi)|\ell\rangle\langle\ell|\partial_{j}\rho(\phi)|k\rangle,\label{eq:WY_def}
\end{equation}
where $\rho(\phi)$ is expressed in terms of its spectral decomposition, as defined in~\eqref{eq:spcetr_decomp}.

To derive our results, we leverage an explicit con\-ne\-ction between the Wigner--Yanase
information of a parameterized family $\left(  \sigma(\gamma)\right)
_{\gamma\in\mathbb{R}^{L}}$ and the Fisher--Bures information of the corresponding pure
parameterized fa\-mi\-ly $\left(  \varphi^{\sigma}(\gamma)\right)  _{\gamma
\in\mathbb{R}^{L}}$. Here, $\varphi^{\sigma}(\gamma)$ denotes a canonical
purification of $\sigma(\gamma)$, defined as
\begin{equation}
\varphi^{\sigma}(\gamma)\coloneqq \left(  \sqrt{\sigma(\gamma)}\otimes
I\right)  \Gamma\left(  \sqrt{\sigma(\gamma)}\otimes I\right)
,\label{eq:canonical-purification-FB-WY}
\end{equation}
where $\Gamma$ is the standard maximally entangled operator:
\begin{equation}\label{eq:max_ent_vec}
\Gamma\coloneqq \sum_{k,\ell}|k\rangle\!\langle\ell|\otimes|k\rangle\!\langle\ell|.
\end{equation}
This relationship is formalized in the following proposition~\cite[Proposition~9]{minervini2025eqbm}:

\begin{proposition}[\cite{minervini2025eqbm}]
\label{prop:FB-WY-canonical-purifications}
For a parameterized, differentiable fa\-mi\-ly $\left(  \sigma(\gamma)\right)
_{\gamma\in\mathbb{R}^{L}}$ of positive-definite states, the following
equality holds:
\begin{equation}
I^{\operatorname{WY}}\!\left(  \gamma;\left(  \sigma(\gamma)\right)
_{\gamma\in\mathbb{R}^{L}}\right)  =I^{\operatorname{FB}}\!\left(
\gamma;\left(  \varphi^{\sigma}(\gamma)\right)  _{\gamma\in\mathbb{R}^{L}
}\right)  ,
\end{equation}
where $\varphi^{\sigma}(\gamma)$ is defined in~\eqref{eq:canonical-purification-FB-WY}.
\end{proposition}

To apply Proposition~\ref{prop:FB-WY-canonical-purifications}, we consider the canonical purification of the parameterized state in~\eqref{eq:state}, given by
\begin{equation}
|\psi(\phi)\rangle\coloneqq \left(  \sqrt{\rho(\phi)}\otimes
I\right)  |\Gamma\rangle,\label{eq:purified-state}
\end{equation}
where $\ket{\Gamma}$ is the  maximally entangled vector defined in~\eqref{eq:max_ent_vec}. Notably, $|\psi(\phi)\rangle$ purifies $\rho(\phi)$ since
\begin{equation}
    \rho(\phi) = \operatorname{Tr}_2[|\psi(\phi)\rangle\!\langle\psi(\phi)|].
\end{equation}

To establish our results, we require the partial derivatives of~\eqref{eq:purified-state}, given by the following lemma:
\begin{lemma}\label{thm:WY-partial-derivatives}
The partial derivatives of the purified state defined in~\eqref{eq:purified-state} satisfy
    \begin{align}
    |\partial_{j}\psi(\phi)\rangle & = i \!\left( \!\left[ \sqrt{\rho(\phi)}, \mathcal{U}_{L_{j+1}}\!(H_j) \right] \! \otimes \! I\!\right)  |\Gamma\rangle,\label{eq:part_der_phi_prufied-evolved-QBM}
\end{align}
where we used the shorthand $\partial_{j} \equiv \frac{\partial}{\partial \phi_j}$.
\end{lemma}
\begin{proof}
    See Appendix~\ref{app:gradient_purified_state}.
\end{proof}
\medskip

\begin{figure*}
    \centering
    \scalebox{1.29}{
    \Qcircuit @C=0.9em @R=1.0em {
        \lstick{\ket{1}} & \gate{\operatorname{Had}} & \gate{S} & \ctrl{2} & \qw & \qw & \qw & \gate{\operatorname{Had}} & \meter &\rstick{\hspace{-1.2em}Z} \\
        \lstick{\ket{1}} & \gate{\operatorname{Had}} & \gate{S} & \qw & \qw & \qw & \ctrl{1} & \gate{\operatorname{Had}} & \meter & \rstick{\hspace{-1.2em}Z}\\
        \lstick{\rho} & \gate{\vphantom{U^\dagger_{R_j}}U_{R_i}} & \qw & \gate{\vphantom{U^\dagger_{R_j}}H_{i}} & \gate{U^\dagger_{R_i}} & \gate{\vphantom{U^\dagger_{R_j}}U_{R_j}} & \gate{\vphantom{U^\dagger_{R_j}}H_{j}} & \gate{U^\dagger_{R_j}} & \meter & \rstick{\hspace{-1.2em}G}
    }
    }
    \caption{Quantum circuit that realizes an unbiased estimate of $ \frac{1}{4} \left\langle \left[ \left[ \mathcal{U}^\dag_{R_j}\!(H_j) , G \right] , \mathcal{U}^\dag_{R_i}\!(H_i) \right] \right\rangle_\rho$, which corresponds to the Kubo--Mori information matrix elements up to a factor of $\frac{1}{4}$. For details of the algorithm, see Appendix~\ref{app:KM_algo}.}
    \label{fig:KM}
\end{figure*}
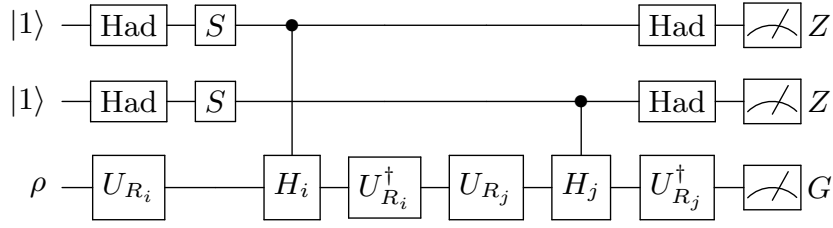

With this result, we can now express the elements of the Wigner–Yanase information matrix as follows:
\begin{theorem}\label{thm:WY}
    Let $\rho$ be an initial thermal state as defined in~\eqref{eq:thermal_state}. For a general layered parameterized quantum circuit of the form in~\eqref{eq:param_circ}, the Wigner--Yanase information matrix elements are given by:
    \begin{align}\label{eq:WY}
    \begin{split}
        I^{\operatorname{WY}}_{ij}(\phi) & = 4 \left\langle \left\{ \mathcal{U}^\dag_{R_i}\!(H_i) , \mathcal{U}^\dag_{R_j}\!(H_j) \right\}\right\rangle_\rho\\
        & \hspace{0.2cm}  - 8 \Tr\!\left[  \mathcal{U}^\dag_{R_i}\!(H_i) \sqrt{\rho} \ \mathcal{U}^\dag_{R_j}\!(H_j) \sqrt{\rho} \right] ,
    \end{split}
    \end{align}
    where $\mathcal{U}^\dagger_{R_j}\!$ denotes the unitary channel defined in~\eqref{eq:q_channel_right}.
\end{theorem}
\begin{proof}
    See Appendix~\ref{app:WY-proof}.
\end{proof}
\medskip

We observe that the expression in~\eqref{eq:WY} closely resembles the result derived in \cite[Theorem~15]{minervini2025eqbm} for the Wigner–Yanase information matrix of quantum evolution machines.  The only distinction lies again in the use of the quantum channel $\mathcal{U}^\dag_{R_i}$ for PQCs. For a detailed discussion of the relationship between quantum evolution machines and PQCs, see Appendix~\ref{app:QEM_PQC}.

The quantity in~\eqref{eq:WY} can be estimated on a quantum computer employing the quantum circuit depicted in Figure~\ref{fig:WY-1} for the first term and the quantum circuit in Figure~\ref{fig:WY-2} for the second term. Detailed procedures for evaluating each of these terms are provided in Appendix~\ref{app:WY-first_term} and \ref{app:WY-second_term}, respectively. Notably, estimating the second term requires access to the canonical purification of the parameterized state in~\eqref{eq:state}, as defined in~\eqref{eq:purified-state}. %\cite{Fang2020no-go_purification}

\begin{remark}
    The Wigner–Yanase information matrix e\-le\-men\-ts, given in~\eqref{eq:WY}, can be equivalently expressed as
    \begin{align}
        I^{\operatorname{WY}}_{ij}(\phi) & = - \, 4 \Tr\!\bigg[  \left[ \mathcal{U}^\dag_{R_i}\!(H_i), \sqrt{\rho}\right] \left[ \mathcal{U}^\dag_{R_j}\!(H_j), \sqrt{\rho}\right] \bigg].
    \end{align}
    Notably, this alternative form is consistent with the common expression for the Wigner–Yanase metric found in the literature~\cite{Gibilisco_2003,Hansen_2008}.
\end{remark}

\subsection{Kubo--Mori information matrix of PQCs with thermal-state initialization}\label{sec:KM}

The main result of this section, Theorem~\ref{thm:KM}, provides an explicit analytical expression for the Kubo--Mori information matrix elements of PQCs initialized with a thermal state. Additionally, we present a quantum circuit designed to estimate the terms in this expression (see Figure~\ref{fig:KM}). Before stating the theorem, let us recall that the Kubo–Mori information matrix elements are given by~\cite[Theorem~7]{minervini2025eqbm}
\begin{equation}
    I^{\operatorname{KM}}_{ij}(\phi) =\sum_{k,\ell}\frac{\ln\lambda_{k}-\ln\lambda_{\ell}}{\lambda_{k}-\lambda_{\ell}}\langle k|\partial_{i}\rho(\phi)|\ell\rangle\langle\ell|\partial_{j}\rho(\phi)|k\rangle,\label{eq:KM_def}
\end{equation}
where $\rho(\phi)$ is expressed in terms of its spectral decomposition, as defined in~\eqref{eq:spcetr_decomp}.

\begin{theorem}\label{thm:KM}
    Let $\rho$ be an initial thermal state as defined in~\eqref{eq:thermal_state}. For a general layered parameterized quantum circuit of the form in~\eqref{eq:param_circ}, the Kubo--Mori information matrix elements are given by:
    \begin{align}
        I^{\operatorname{KM}}_{ij}(\phi) = \left\langle \left[ \left[ \mathcal{U}^\dag_{R_j}\!(H_j) , G \right] \!, \mathcal{U}^\dag_{R_i}\!(H_i) \right] \right\rangle_\rho,\label{eq:KM}
    \end{align}
    where $\mathcal{U}^\dagger_{R_j}\!$ denotes the unitary channel defined in~\eqref{eq:q_channel_right}.
\end{theorem}
\begin{proof}
    See Appendix~\ref{app:KM-proof}.
\end{proof}
\medskip

Once again, we observe a similarity with the result in \cite[Theorem~19]{minervini2025eqbm} for quantum evolution machines. For details on the relationship between quantum evolution machines and PQCs, see Appendix~\ref{app:QEM_PQC}.

The quantity in~\eqref{eq:KM} can be estimated on a quantum computer by employing the quantum circuit shown in Figure~\ref{fig:KM}. A detailed procedure for its estimation is provided in Appendix~\ref{app:KM_algo}.

\section{Quantum natural gradient for PQCs with thermal-state initialization}\label{sec:natural_gradient}

In variational quantum algorithms, a typical cost function is defined as
\begin{equation}\label{eq:cost_function}
    C(\phi) \coloneqq \Tr[O \rho(\phi)],
\end{equation}
where $O$ is a Hermitian operator and $\rho(\phi)$ is a parameterized quantum state. Here, we consider the parameterized quantum state defined in~\eqref{eq:state}, which is generated by the parameterized quantum circuit in~\eqref{eq:param_circ}. 

Optimizing the cost function requires efficient strategies that appro\-xi\-mate the optimal solution within a reasonable number of iterations. A variety of optimization methods have been explored in the context of variational quantum algorithms, drawing inspiration from classical machine learning techniques~\cite{Sung2020optimizersvqa,Pellow_Jarman2021opt_pqc,Lavrijsen2020opt_pqc}. These optimizers can be categorized based on the order of information they require about the cost function. Zeroth-order methods, also known as direct optimization methods, rely solely on cost function evaluations~\cite{Nelder1965nelder_mead,Powell1964powell_method}. First-order methods require access to the gradient~\cite{ruder2017overviewgradientdescentoptimization,kingma2017adam}, while second-order methods utilize the Hessian or other metrics reflecting the local curvature of the optimization landscape~\cite{Martens2020}. Among these, considering that direct optimization does not scale efficiently to large problem sizes, gradient-based methods are currently the most widely used in quantum machine learning. The simplest gradient-based algorithm is gradient descent, which updates the parameter vector $\phi$ based on the gradient of the cost function and a fixed learning rate~$\mu >0$:
\begin{equation}\label{eq:grad_desc}
    \phi_{m+1} = \phi_m - \mu \nabla_\phi C(\phi_m).
\end{equation}
Thus, evaluating the gradient of the cost function~\eqref{eq:cost_function} is essential. Using~\eqref{eq:grad_pqc}, the gradient can be computed as
\begin{align}
   \frac{\partial}{\partial\phi_{j}}  C(\phi) &= \Tr\!\left[O \frac{\partial}{\partial\phi_{j}}\rho(\phi)\right]\\
    & = i \Tr\! \left[ O \left[ \rho(\phi) , U_{L_j} H_j U^\dagger_{L_j} \right] \right]\\
    &= i \Tr\! \left[ \left[  U_{L_j} H_j U^\dagger_{L_j} , O \right] \rho(\phi) \right]\\
    &= i \Tr\! \left[ \left[ \mathcal{U}_{L_j}\!(H_j) , O \right] \rho(\phi) \right].
\end{align}
The structure of parameterized quantum circuits in~\eqref{eq:param_circ} allows for efficient gradient evaluation on a quantum computer. Gradients can be computed using measurements of rotation generators, such as the parameter shift rule and finite difference methods~\cite{Schuld2019gradient_evaluation,Izmaylov2021param_shift_rule_general,Hubregtsen2022grad_pqc,Wierichs2022generalparameter,Banchi2021stoch_par_shift_rule}, or through ancilla-based quantum gradient estimation schemes~\cite{romero2018,li2024efficientquantumgradienthigherorder}.

Since the gradient points in the direction of stee\-pest ascent in the Euclidean geometry of the parameter space, gradient-based optimization updates parameters toward the steepest descent of the cost function within this framework, intuitively leading to convergence at a minimum. However, Euclidean geometry is typically not well suited for the space of quantum states~\cite{Stokes2020qng}. Consequently, while the optimizer in~\eqref{eq:grad_desc} detects convergence through a vanishing gradient, it cannot distinguish between local and global minima. Furthermore, the trai\-ning landscapes of variational quantum algorithms often exhibit numerous local minima that are far from optimal~\cite{Bittel2021vqaNPhard}. As a result, both gradient-based and higher-order descent algorithms typically converge to suboptimal solutions. 

To address these limitations, natural gradient descent provides a more judicious approach to navigating the optimization landscape and mi\-ti\-ga\-ting the risk of getting trapped in local minima~\cite{Wierichs2020nat_grad_avoid_local_min}. Originally introduced in the classical setting~\cite{amari1998nat_grad}, natural gradient descent differs from standard gradient descent by incorporating the information geometry of the parameterized space, rather than relying on Euclidean geometry. This is achieved through the Fisher information matrix, which ensures parameter updates follow the steepest direction in the state space. The algorithm has been extended to the quantum setting by incorporating quantum generalizations of the Fisher information matrix, which are derived from feasible smooth divergences of quantum states~\cite{Stokes2020qng,sbahi2022provablyefficientvariationalgenerative,Koczor2022qng_noisy,sohail2024quantumnaturalstochasticpairwise,patel2024naturalgradientparameterestimation,minervini2025eqbm}. 

The update rule of the quantum natural gradient descent algorithm is given by
\begin{equation}\label{eq:nat_grad}
    \phi_{m+1} = \phi_m - \mu \left[ I^{\boldsymbol{D}} (\phi_m) \right]^{-1} \nabla_\phi C(\phi_m),
\end{equation}
where $I^{\boldsymbol{D}}(\gamma_m)$ is the $\boldsymbol{D}$-based Fisher information matrix, defined in~\eqref{eq:D-based-Fisher-matrix}, which encodes the curvature of the parameter space, and $\mu > 0$ is the learning rate. By incorporating the inverse of $I^{\boldsymbol{D}}(\phi_m)$, the gradient is rescaled to align with the steepest descent direction in the geometry defined by the smooth divergence $\boldsymbol{D}$, rather than the Euclidean metric. This rescaling improves optimization efficiency, helping to avoid local minima and enhancing overall convergence. A convergence analysis of quantum natural gradient under certain assumptions has been presented in~\cite{sohail2024quantumnaturalstochasticpairwise}.

Applying natural gradient descent to PQCs requires estimating the relevant information matrices and computing their inverses. While this additional computation introduces some overhead, it is often compensated by improved convergence rates and a reduced number of iterations, as well as the ability to escape saddle points more effectively~\cite{Straaten2021cost_nat_grad}.
In Section~\ref{sec:infor_matrices}, we outlined methods for evaluating the Fisher--Bures $I^{\operatorname{FB}}(\phi)$ (Section~\ref{sec:FB}), Wigner--Yanase $I^{\operatorname{WY}}(\phi)$ (Section~\ref{sec:WY}), and Kubo--Mori $I^{\operatorname{KM}}(\phi)$ (Section~\ref{sec:KM}) information matrix elements for PQCs. Each of these quantum generalizations of the Fisher information matrix can be directly incorporated into the quantum natural gradient descent algorithm by substituting the chosen matrix into the update rule defined in~\eqref{eq:nat_grad}. This additional flexibility allows for customization of the optimization process to fit specific problem settings. In this regard, \cite[Corollary 8]{minervini2025eqbm} establishes that, in the matrix Loewner order, the Fisher–Bures and Wigner–Yanase information matrices differ by at most a factor of two. As a result, the optimization steps derived from these two matrices differ only by a constant factor, making them essentially interchangeable in practical implementations of quantum natural gradient descent. This flexibility is particularly advantageous when one metric is computationally easier to estimate than the other.

\section{Parameter estimation of states generated by PQCs}

In addition to their direct application in quantum natural gradient descent, our results on the quantum generalizations of the Fisher information matrix are also relevant to parameter estimation (see~\cite{Liu2019q_Fisher_multi_par_est,Sidhu2020} for recent reviews). In this context, the Fisher--Bures information matrix provides a fundamental limit on how precisely unknown parameters can be estimated, as described by the following multiparameter Cramer--Rao bound, which holds for an arbitrary unbiased estimator and for a general parameterized family $(\sigma(\gamma))_{\gamma\in\mathbb{R}^{L}
}$ of states:
\begin{equation}
\operatorname{Cov}^{(n)}(\hat{\gamma},\gamma)\geq\frac{1}{n}\left[  I^{\operatorname{FB}
}(\gamma)\right]  ^{-1}.
\label{eq:Cramer--Rao-multiple}
\end{equation}
Here, $n\in\mathbb{N}$ is the number of copies of the state $\sigma(\gamma)$
available, $\hat{\gamma}$ is an estimate of the parameter vector~$\gamma$, the matrix $I^{\operatorname{FB}}(\gamma)$ denotes the Fisher--Bures
information matrix, and the covariance matrix
$\operatorname{Cov}^{(n)}(\hat{\gamma},\gamma)$ measures errors in estimation and is
defined in terms of its matrix elements as
\begin{multline}
\lbrack\operatorname{Cov}^{(n)}(\hat{\gamma},\gamma)]_{k,\ell}\coloneqq\\
\sum_{m}\operatorname{Tr}[M_{m}^{(n)}\sigma(\gamma)^{\otimes n}](\hat{\gamma
}_{k}(m)-\gamma_{k})(\hat{\gamma}_{\ell}(m)-\gamma_{\ell}).
\label{eq:cov-mat-def}
\end{multline}
In~\eqref{eq:cov-mat-def}, $(M_{m}^{(n)})_{m}$ is an arbitrary positive operator-valued measure used for estimation, i.e.,
satisfying $M_{m}^{(n)}\geq0$ for all $m$ and $\sum_{m}M_{m}^{(n)}=I^{\otimes n}$. This measurement acts, in general, collectively on all $n$ copies of the
state $\sigma(\gamma)^{\otimes n}$. Additionally,
\begin{equation}
\hat{\gamma}(m)\coloneqq(\hat{\gamma}_{1}(m),\hat{\gamma}_{2}(m),\ldots
,\hat{\gamma}_{J}(m))^{\mathsf{T}}
\label{eq:parameter-estimate-function}
\end{equation}
is a function that maps the measurement outcome $m$ to an estimate
$\hat{\gamma}(m)$ of the parameter vector $\gamma$. The inequality in~\eqref{eq:Cramer--Rao-multiple} exploits the additivity of the Fisher--Bures information matrix, as reviewed in~\cite[Appendix~A]{patel2024naturalgradientparameterestimation}. As noted in~\cite[Eq.~(C11)]{sbahi2022provablyefficientvariationalgenerative}, the multiparameter Cramer--Rao bound in~\eqref{eq:Cramer--Rao-multiple} can be written as follows:
\begin{equation}
\begin{pmatrix}
    \operatorname{Cov}^{(n)}(\hat{\gamma},\gamma) & I \\
    I & n  I^{\operatorname{FB}}(\gamma)
\end{pmatrix} \geq 0,
\end{equation}
which is a direct consequence of the Schur complement lemma.

The Fisher--Bures, Wigner--Yanase, and Kubo--Mori information matrices are related through the following operator inequalities (see~\cite[Section~III-D-1]{Jarzyna2020} and~\cite[Corollary 8]{minervini2025eqbm}):
\begin{align}
I^{\operatorname{KM}}(\gamma) \geq I^{\operatorname{WY}}(\gamma) \geq I^{\operatorname{FB}}(\gamma) \geq \frac{1}{2} I^{\operatorname{WY}}(\gamma).
\end{align}
From these inequalities and the multiparameter Cramer--Rao bound in~\eqref{eq:Cramer--Rao-multiple}, we see that the Fisher--Bures and Wigner--Yanase information matrices can be used interchangeably in the low-error regime, as they provide similar bounds on the covariance matrix $\operatorname{Cov}^{(n)}(\hat{\gamma},\gamma)$. Both matrices also give tighter bounds than the Kubo--Mori information matrix. This interchangeability is useful in the case that the Fisher--Bures information matrix elements are difficult to evaluate but the Wigner--Yanase
information matrix elements are not.

Now let us consider a scenario where a quantum state $\rho(\phi)$ is generated by an unknown PQC with unknown parameter vector $\phi$. The task is to estimate the unknown parameter vector $\phi$. The Fisher--Bures information matrix provides a natural framework for this task, as it quantifies how sensitive $\rho(\phi)$ is to changes in $\phi$. By computing the Fisher--Bures information matrix, we obtain fundamental limits on the performance of an arbitrary scheme for estimating the parameter vector $\phi$ from the results of measurements performed on $\rho(\phi)$, as described by the multiparameter Cramer--Rao bound in~\eqref{eq:Cramer--Rao-multiple}.
Thus, in the context of estimating the parameters of a state generated by a PQC, our unbiased expressions and quantum algorithms for computing the Fisher--Bures and Wigner--Yanase information matrices can guide the search for the unknown PQC parameters. By evaluating the information matrices for different candidate PQCs, one can assess how sensitive each circuit is to changes in its parameters. This information can be used to iteratively refine the PQC, by adjusting the parameter vector $\phi$.
%or by modifying the circuit structure itself. %Additionally, the quantum Fisher information matrix can be used to compare different PQCs and select the one that offers the best trade-off between expressivity and trainability.

% \mm{is this relevant?}
% \mmw{I think we have to eliminate this paragraph} Another possible application of the quantum Fisher information matrix is in the robustness and sensitivity analysis of PQCs. Noise and imperfections can significantly affect the performance of PQCs. By computing the quantum Fisher information matrix, one can quantify how sensitive a PQC is to variations in its parameters, which is crucial for assessing the robustness of the circuit to noise. By analyzing the quantum Fisher information matrix, one can identify which parameters require precise control and which can tolerate more significant variations, leading to more robust circuit designs. 

\section{Conclusion and future directions}

In this paper, we have characterized and developed quantum algorithms to compute three quantum generalizations of the Fisher information matrix, namely Fisher--Bures, Wigner--Yanase, and Kubo--Mori, for PQCs initialized with thermal states. 
These algorithms require access to either the Hamiltonian that prepares the initial thermal state or its canonical purification. Specifically, for the Hamiltonian, the access is semi-classical: we assume knowledge of its description as a local Hamiltonian, which can be efficiently represented classically and measured on a quantum computer. This is a reasonable assumption in many practical scenarios, as Hamiltonians in quantum systems are often known and can be decomposed into local terms. On the other hand, access to the canonical purification of the initial state involves creating and manipulating a maximally entangled state between the system and an auxiliary register, which is a quantum task requiring specific state-preparation techniques. The distinction between these two types of access highlights the trade-offs in resource requirements for implementing our algorithms: the Hamiltonian-based approach leverages classical descriptions and requires quantum control for measurements, while the purification-based approach necessitates quantum control and state preparation capabilities.

Despite these requirements, our results represent the first exact expressions for these quantum information matrices that can be computed on a quantum computer, advancing beyond previous approximation methods for mixed-state scenarios and addressing an open problem regarding the geometry of mixed-state PQCs~\cite{Stokes2020qng}. A summary of our findings is provided in Table~\ref{table:FB-WY-KM-results}, for which we highlight the similarity with the information matrix elements of quantum evolution machines, as shown in~\cite[Table~1]{minervini2025eqbm}. The key difference between information matrices of PQCs and quantum evolution machines lies in the quantum channels arising from their respective gradients, which reflect the distinct structures of the two ansatzes. Moreover, the thermal state initialization in our framework aligns with the construction of Quantum Boltzmann Machines (QBMs). In this context, employing PQCs on top of QBMs provides an alternative approach for leveraging potential trainability advantages of QBMs~\cite{coopmans2024qbm,patel2024quantumboltzmannmachine} while retaining the flexibility and ease of use of PQCs.

By deriving three quantum generalizations of the Fisher information matrix, our work enables the implementation of three distinct quantum natural gradient descent algorithms for PQCs initialized with a thermal state. The availability of multiple quantum information matrices allows for flexibility in choosing the most sui\-ta\-ble one for a given optimization problem. Furthermore, as established in~\cite[Corollary 8]{minervini2025eqbm}, the Fisher–Bures and Wigner–Yanase information matrices differ by no more than a factor of two in the matrix Loewner order, making them essentially interchangeable in practical implementations of natural gradient descent.

Beyond optimization, our results also establish fundamental limits on the precision with which the parameters of a PQC can be estimated. Through the quantum Cramer--Rao bound, the quantum generalizations of the Fisher information matrix derived in this work quantify the sensitivity of the quantum state to changes in the circuit parameters, providing fundamental limitations on how well one can estimate the unknown parameters of a state generated by a PQC.

Looking ahead, while our work establishes a detailed theoretical foundation for the Fisher–Bures, Wigner–Yanase, and Kubo–Mori information matrices, several important questions remain open. First, the question of which matrix is best suited for specific applications remains open, as discussed in~\cite{minervini2025eqbm}. Addressing this question could involve numerical simulations to evaluate the performance of the corresponding quantum natural gradient algorithms on benchmark problems, including comparisons with standard gradient descent and other optimization methods. Such simulations would offer insights into convergence rates, scalability, and robustness, helping to identify the most effective matrix for different scenarios. In future work, we plan to explore these directions, building on the theoretical developments presented here and the simplifying methods from \cite{Gacon2021stoch_approx_Fisher,Fitzek2024optimizing,sohail2024quantumnaturalstochasticpairwise,halla2025estimationquantumfisherinformation,halla2025modifiedconjugatequantumnatural}. 

Furthermore, a key focus will be to investigate whether these novel optimization methods can address critical challenges in training PQCs, such as avoiding local minima and escaping barren plateaus. Indeed, while it is well known that, under certain conditions, the variance of the cost function gradient for PQCs decays exponentially with system size, it has yet to be proven whether the \textit{natural gradient}---obtained by rescaling the gradient with the inverse of the quantum Fisher information matrix---exhibits the same exponential decay. By tackling these issues, our work has the potential to enhance the scalability and efficiency of PQCs in variational quantum algorithms, paving the way for more robust quantum computing applications.

While our work provides exact formulas and unbiased estimators for these information matrices, practical implementations face significant challenges. The resource overhead of estimating and inverting these matrices, particularly in high-dimensional settings, requires careful analysis of sample complexity and computational costs. Moreover, the inversion of noisy matrix estimates introduces numerical instability and bias, which may affect both optimization performance and the validity of the Cramér--Rao bound. Potential mitigation strategies include regularization techniques, iterative linear system solvers, and averaging multiple independent estimates, though a rigorous analysis of these approaches remains an important direction for future work.

Finally, our results assume full-rank thermal states, yet in practice one may encounter low-rank density operators or projected Gibbs states $P\rho P$. An open question is whether proximity in trace norm implies proximity of the corresponding information matrices.  Specifically, given a rank-$k$ state $\rho$ and full-rank state $\rho_G$ with Hamiltonian $G$, an open question is to determine whether $\|\rho-\rho_G\|_1\le\varepsilon$ implies that one can bound $\bigl|I_{ij}(\phi)-I^{G}_{ij}(\phi)\bigr|\le f_{ij}(\varepsilon)$.  
Existing analytic results~\cite{liu2016,rezakhani2019} for non-full-rank states suggest that such continuity bounds should scale linearly with $\varepsilon$ for small perturbations. Establishing precise inequalities of this type and extending our unbiased-estimation schemes to arbitrary rank-$k$ or projected thermal states represent an interesting direction for future research.

\section*{Acknowledgments}

We thank Zoe Holmes and Soorya Rethinasamy for helpful discussions. 

MM and DP acknowledge support from the Air Force Office of Scientific Research under agreement no. FA2386-24-1-4069. DP and MMW acknowledge support from Air Force Research Laboratory under agreement no. FA8750-23-2-0031.

The U.S. Government is authorized to reproduce and distribute reprints for Governmental purposes notwithstanding any copyright notation thereon. The views and conclusions contained herein are those of the authors and should not be interpreted as necessarily representing the official policies or endorsements, either expressed or implied, of Air Force Research Laboratory or the United States Air Force or the U.S. Government.

\section*{Author Contributions}
\noindent
\textbf{Author Contributions}: The following describes the
different contributions of the authors of this work, using
roles defined by the CRediT (Contributor Roles Taxonomy) project~\cite{CRediT}:

\noindent\textbf{MM:} Formal analysis, Investigation, Methodology, Validation, Writing - original draft, Writing - review \& editing.

\noindent\textbf{DP:} Formal analysis, Validation, Writing - review \& editing.

\noindent\textbf{MMW:} Conceptualization, Funding acquisition, Supervision, Validation, Writing - review \& editing.

\bibliography{references}

\appendix
\onecolumngrid 

\section{Gradient of PQCs (Proof of Lemma~\ref{lem:part_der})}
\label{app:grad_pqc}

Using the shorthand $\partial_{j}\equiv\frac{\partial}{\partial\phi_{j}}$, consider that 
\begin{align}
    \partial_j \rho(\phi) &= (\partial_j U(\phi)) \rho U^\dagger(\phi) +U(\phi) \rho \ (\partial_j U^\dagger(\phi))\\
    &= -i U_{L_{j+1}} H_j U_{R_j} \rho U^\dagger(\phi) + i U(\phi) \rho U^\dagger_{R_j} H_j U^\dagger_{L_{j+1}}\\
    &= -i U_{L_{j+1}} H_j U^\dagger_{L_{j+1}} U_{L_{j+1}} U_{R_j} \rho U^\dagger(\phi) + i U(\phi) \rho U^\dagger_{R_j} U^\dagger_{L_{j+1}} U_{L_{j+1}} H_j U^\dagger_{L_{j+1}}\\
    &= -i U_{L_{j+1}} H_j U^\dagger_{L_{j+1}} U(\phi) \rho U^\dagger(\phi)  + i U(\phi) \rho U^\dagger(\phi) U_{L_{j+1}} H_j U^\dagger_{L_{j+1}}\\
    & = i \left[ U(\phi) \rho U^\dagger(\phi) , U_{L_{j+1}} H_j U^\dagger_{L_{j+1}} \right]\\
    & = i \left[ \rho(\phi) , U_{L_{j+1}} H_j U^\dagger_{L_{j+1}} \right]\\
    & =  i \left[ \rho(\phi) , \mathcal{U}_{L_{j+1}}\!(H_j) \right], \label{step:grad_rho-2}
\end{align}
where in~\eqref{step:grad_rho-2} we have used the shorthand introduced in~\eqref{eq:q_channel_left}.

\section{Fisher--Bures information matrix elements}\label{app:FB}

\subsection{Proof of Theorem~\ref{thm:FB}}\label{app:FB-proof}
The derivation follows a similar approach to that of \cite[Theorem~11]{minervini2025eqbm}, but with some distinctions. Therefore, we provide a detailed proof below. Using the result from Lemma~\ref{lem:matr_part_der}, consider that
\begin{align}
    I^{\operatorname{FB}}_{ij}(\phi) &= \sum_{k,\ell}\frac{2}{\lambda_{k}+\lambda_{\ell}} \langle {k}| \partial_{i}\rho(\phi) |\ell\rangle \langle\ell| \partial_{j}\rho(\phi) |{k}\rangle\\
    & = \sum_{k,\ell}\frac{2}{\lambda_{k}+\lambda_{\ell}} \left( i(\lambda_k - \lambda_\ell) \langle \tilde{k}| \mathcal{U}^\dag_{R_i}\!(H_i) |\tilde{\ell}\rangle \right) \left( i(\lambda_\ell - \lambda_k) \langle\tilde{\ell}| \mathcal{U}^\dag_{R_j}\!(H_j) |\tilde{k}\rangle \right)\\
    & = 2 \sum_{k,\ell}\frac{(\lambda_k - \lambda_\ell)^2}{\lambda_{k}+\lambda_{\ell}} \langle \tilde{k}| \mathcal{U}^\dag_{R_i}\!(H_i) |\tilde{\ell}\rangle \langle\tilde{\ell}| \mathcal{U}^\dag_{R_j}\!(H_j) |\tilde{k}\rangle\\
    & = 2 \sum_{k,\ell}\frac{(\lambda_k + \lambda_\ell)^2 - 4\lambda_k \lambda_\ell}{\lambda_{k}+\lambda_{\ell}} \langle \tilde{k}| \mathcal{U}^\dag_{R_i}\!(H_i) |\tilde{\ell}\rangle \langle\tilde{\ell}| \mathcal{U}^\dag_{R_j}\!(H_j) |\tilde{k}\rangle\\
    \begin{split}
    & = 2 \sum_{k,\ell}(\lambda_k + \lambda_\ell) \langle \tilde{k}| \mathcal{U}^\dag_{R_i}\!(H_i) |\tilde{\ell}\rangle \langle\tilde{\ell}| \mathcal{U}^\dag_{R_j}\!(H_j) |\tilde{k}\rangle \\
    & \hspace{1cm} - 8 \sum_{k,\ell}\frac{\lambda_k \lambda_\ell}{\lambda_{k}+\lambda_{\ell}} \langle \tilde{k}| \mathcal{U}^\dag_{R_i}\!(H_i) |\tilde{\ell}\rangle \langle\tilde{\ell}| \mathcal{U}^\dag_{R_j}\!(H_j) |\tilde{k}\rangle
    \end{split}\\
    \begin{split}& = 2 \sum_{k,\ell}\lambda_k \langle \tilde{k}| \mathcal{U}^\dag_{R_i}\!(H_i) |\tilde{\ell}\rangle \langle\tilde{\ell}| \mathcal{U}^\dag_{R_j}\!(H_j) |\tilde{k}\rangle + 2 \sum_{k,\ell}\lambda_\ell \langle \tilde{k}| \mathcal{U}^\dag_{R_i}\!(H_i) |\tilde{\ell}\rangle \langle\tilde{\ell}| \mathcal{U}^\dag_{R_j}\!(H_j) |\tilde{k}\rangle\\
    & \hspace{1cm} - 8 \sum_{k,\ell}\frac{\lambda_k \lambda_\ell}{\lambda_{k}+\lambda_{\ell}} \langle \tilde{k}| \mathcal{U}^\dag_{R_i}\!(H_i) |\tilde{\ell}\rangle \langle\tilde{\ell}| \mathcal{U}^\dag_{R_j}\!(H_j) |\tilde{k}\rangle
    \end{split}\\
    \begin{split}& = 2 \Tr\!\left[ \mathcal{U}^\dag_{R_j}\!(H_j) \rho \mathcal{U}^\dag_{R_i}\!(H_i) \right] + 2 \Tr\!\left[ \mathcal{U}^\dag_{R_i}\!(H_i) \rho \mathcal{U}^\dag_{R_j}\!(H_j) \right]\\
    & \hspace{1cm} - 8 \sum_{k,\ell}\frac{\lambda_k \lambda_\ell}{\lambda_{k}+\lambda_{\ell}} \langle \tilde{k}| \mathcal{U}^\dag_{R_i}\!(H_i) |\tilde{\ell}\rangle \langle\tilde{\ell}| \mathcal{U}^\dag_{R_j}\!(H_j) |\tilde{k}\rangle
    \end{split}\\
    \begin{split}& = 2 \Tr\!\left[\mathcal{U}^\dag_{R_i}\!(H_i) \mathcal{U}^\dag_{R_j}\!(H_j) \rho  \right] + 2 \Tr\!\left[ \mathcal{U}^\dag_{R_j}\!(H_j) \mathcal{U}^\dag_{R_i}\!(H_i) \rho \right]\\
    & \hspace{1cm} - 8 \sum_{k,\ell}\frac{\lambda_k \lambda_\ell}{\lambda_{k}+\lambda_{\ell}} \langle \tilde{k}| \mathcal{U}^\dag_{R_i}\!(H_i) |\tilde{\ell}\rangle \langle\tilde{\ell}| \mathcal{U}^\dag_{R_j}\!(H_j) |\tilde{k}\rangle
    \end{split}\\
    & = 2 \Tr\!\left[ \left\{ \mathcal{U}^\dag_{R_i}\!(H_i), \mathcal{U}^\dag_{R_j}\!(H_j) \right\} \rho  \right] - 8 \sum_{k,\ell}\frac{\lambda_k \lambda_\ell}{\lambda_{k}+\lambda_{\ell}} \langle \tilde{k}| \mathcal{U}^\dag_{R_i}\!(H_i) |\tilde{\ell}\rangle \langle\tilde{\ell}| \mathcal{U}^\dag_{R_j}\!(H_j) |\tilde{k}\rangle\\
    & = 2 \left\langle\left\{ \mathcal{U}^\dag_{R_i}\!(H_i), \mathcal{U}^\dag_{R_j}\!(H_j) \right\}\right\rangle_{\rho} - 8 \sum_{k,\ell}\frac{\lambda_k \lambda_\ell}{\lambda_{k}+\lambda_{\ell}} \langle \tilde{k}| \mathcal{U}^\dag_{R_i}\!(H_i) |\tilde{\ell}\rangle \langle\tilde{\ell}| \mathcal{U}^\dag_{R_j}\!(H_j) |\tilde{k}\rangle.
\end{align}
Now let us focus on the second term:
\begin{align}
    &\sum_{k,\ell}\frac{\lambda_k \lambda_\ell}{\lambda_{k}+\lambda_{\ell}} \langle \tilde{k}| \mathcal{U}^\dag_{R_i}\!(H_i) |\tilde{\ell}\rangle \langle\tilde{\ell}| \mathcal{U}^\dag_{R_j}\!(H_j) |\tilde{k}\rangle\notag\\
    & = \operatorname{Re}\!\left[ \sum_{k,\ell}\frac{\lambda_k \lambda_\ell}{\lambda_{k}+\lambda_{\ell}} \langle \tilde{k}| \mathcal{U}^\dag_{R_i}\!(H_i) |\tilde{\ell}\rangle \langle\tilde{\ell}| \mathcal{U}^\dag_{R_j}\!(H_j) |\tilde{k}\rangle \right]\\
    & = \operatorname{Re}\!\left[ \sum_{k,\ell}\frac{\lambda_k \frac{e^{-\mu_\ell}}{Z}}{\frac{e^{-\mu_k}}{Z}+\frac{e^{-\mu_\ell}}{Z}} \langle \tilde{k}| \mathcal{U}^\dag_{R_i}\!(H_i) |\tilde{\ell}\rangle \langle\tilde{\ell}| \mathcal{U}^\dag_{R_j}\!(H_j) |\tilde{k}\rangle \right]\\
    & = \operatorname{Re}\!\left[ \sum_{k,\ell}\lambda_k \frac{1}{e^{-(\mu_k-\mu_\ell)}+1} \langle \tilde{k}| \mathcal{U}^\dag_{R_i}\!(H_i) |\tilde{\ell}\rangle \langle\tilde{\ell}| \mathcal{U}^\dag_{R_j}\!(H_j) |\tilde{k}\rangle \right].
\end{align}
Now observe that, for all $x\in\mathbb{R}$,
\begin{align}
\frac{1}{e^{-x}+1} &  =\frac{e^{x/2}}{e^{x/2}+e^{-x/2}}\\
&  =\frac{e^{x/2}}{e^{x/2}+e^{-x/2}}-\frac{1}{2}+\frac{1}{2}\\
&  =\frac{e^{x/2}}{e^{x/2}+e^{-x/2}}-\frac{\frac{1}{2}e^{x/2}+\frac{1}
{2}e^{-x/2}}{e^{x/2}+e^{-x/2}}+\frac{1}{2}\\
&  =\frac{1}{2}\frac{e^{x/2}-e^{-x/2}}{e^{x/2}+e^{-x/2}}+\frac{1}{2}\\
&  =\frac{1}{2}\tanh(x/2)+\frac{1}{2}\\
&  =\frac{x}{4}\frac{\tanh(x/2)}{x/2}+\frac{1}{2}.
\end{align}
Substituting above, we find that
\begin{align}
&  \operatorname{Re}\!\left[  \sum_{k,\ell}\lambda_{k}\frac{1}{e^{-\left(
\mu_{k}-\mu_{\ell}\right)  }+1} \langle \tilde{k}| \mathcal{U}^\dag_{R_i}\!(H_i) |\tilde{\ell}\rangle \langle\tilde{\ell}| \mathcal{U}^\dag_{R_j}\!(H_j) |\tilde{k}\rangle \right]  \nonumber\\
&  =\operatorname{Re}\!\left[  \sum_{k,\ell}\lambda_{k}\left(  \frac{\mu_{k}
-\mu_{\ell}}{4}\frac{\tanh(\left(  \mu_{k}-\mu_{\ell}\right)  /2)}{\left(
\mu_{k}-\mu_{\ell}\right)  /2}+\frac{1}{2}\right)  \langle \tilde{k}| \mathcal{U}^\dag_{R_i}\!(H_i) |\tilde{\ell}\rangle \langle\tilde{\ell}| \mathcal{U}^\dag_{R_j}\!(H_j) |\tilde{k}\rangle \right]  \\
\begin{split}&  =\frac{1}{4}\operatorname{Re}\!\left[  \sum_{k,\ell}\lambda_{k}\left(
\mu_{k}-\mu_{\ell}\right)  \frac{\tanh(\left(  \mu_{k}-\mu_{\ell}\right)
/2)}{\left(  \mu_{k}-\mu_{\ell}\right)  /2} \langle \tilde{k}| \mathcal{U}^\dag_{R_i}\!(H_i) |\tilde{\ell}\rangle \langle\tilde{\ell}| \mathcal{U}^\dag_{R_j}\!(H_j) |\tilde{k}\rangle \right]  \\
&  \qquad+\frac{1}{2}\operatorname{Re}\!\left[  \sum_{k,\ell}\lambda_{k}
\langle \tilde{k}| \mathcal{U}^\dag_{R_i}\!(H_i) |\tilde{\ell}\rangle \langle\tilde{\ell}| \mathcal{U}^\dag_{R_j}\!(H_j) |\tilde{k}\rangle \right]  
\end{split}\\
\begin{split}&  =\frac{1}{4}\operatorname{Re}\!\left[  \sum_{k,\ell}\lambda_{k}\left(
\mu_{k}-\mu_{\ell}\right)  \int_{\mathbb{R}}dt\ p(t)e^{-i\left(  \mu_{k}
-\mu_{\ell}\right)  t} \langle \tilde{k}| \mathcal{U}^\dag_{R_i}\!(H_i) |\tilde{\ell}\rangle \langle\tilde{\ell}| \mathcal{U}^\dag_{R_j}\!(H_j) |\tilde{k}\rangle \right]
\\
&  \qquad+\frac{1}{2}\operatorname{Re}\!\left[  \operatorname{Tr}\!\left[ \mathcal{U}^\dag_{R_j}\!(H_j) \rho \mathcal{U}^\dag_{R_i}\!(H_i) \right]
\right]  
\end{split}\\
\begin{split}&  =\frac{1}{4}\operatorname{Re}\!\left[  \sum_{k,\ell}\lambda_{k}\left(
\mu_{k}-\mu_{\ell}\right)  \int_{\mathbb{R}}dt\ p(t)e^{-i\left(  \mu_{k}
-\mu_{\ell}\right)  t} \langle \tilde{k}| \mathcal{U}^\dag_{R_i}\!(H_i) |\tilde{\ell}\rangle \langle\tilde{\ell}| \mathcal{U}^\dag_{R_j}\!(H_j) |\tilde{k}\rangle \right]
\\
&  \qquad+\frac{1}{4}\left\langle \left\{  \mathcal{U}^\dag_{R_i}\!(H_i) , \mathcal{U}^\dag_{R_j}\!(H_j) \right\}  \right\rangle _{\rho}.\end{split}
\end{align}
The third equality above follows from \cite[Lemma~12]{patel2024quantumboltzmannmachine}.
Consider now that the key quantity in the first term above can be rewritten as
\begin{align}
&  \sum_{k,\ell}\lambda_{k}\left(  \mu_{k}-\mu_{\ell}\right)  \int
_{\mathbb{R}}dt\ p(t)e^{-i\left(  \mu_{k}-\mu_{\ell}\right)  t} \langle \tilde{k}| \mathcal{U}^\dag_{R_i}\!(H_i) |\tilde{\ell}\rangle \langle\tilde{\ell}| \mathcal{U}^\dag_{R_j}\!(H_j) |\tilde{k}\rangle \nonumber\\
&  =\int_{\mathbb{R}}dt\ p(t)\sum_{k,\ell}\lambda_{k}\left(  \mu_{k}-\mu
_{\ell}\right)  \operatorname{Tr}\!\left[  \mathcal{U}^\dag_{R_i}\!(H_i) e^{i\mu_{\ell}t}|\tilde{\ell}\rangle\langle\tilde{\ell}| \mathcal{U}^\dag_{R_j}\!(H_j) e^{-i\mu_{k}t}|\tilde{k}\rangle\langle\tilde{k}|\right]  \\
\begin{split} &  =\int_{\mathbb{R}}dt\ p(t)\sum_{k,\ell}\operatorname{Tr}\!\left[  \mathcal{U}^\dag_{R_i}\!(H_i) e^{i\mu_{\ell}t}|\tilde{\ell}\rangle\langle\tilde{\ell}
|\mathcal{U}^\dag_{R_j}\!(H_j)\lambda_{k}\mu_{k}e^{-i\mu_{k}t}|\tilde{k}
\rangle\langle\tilde{k}|\right]  \\
&  \qquad-\int_{\mathbb{R}}dt\ p(t)\sum_{k,\ell}\operatorname{Tr}\!\left[
\mathcal{U}^\dag_{R_i}\!(H_i)  \mu_{\ell}e^{i\mu_{\ell}t}|\tilde{\ell}\rangle
\langle\tilde{\ell}|\mathcal{U}^\dag_{R_j}\!(H_j)\lambda_{k}e^{-i\mu_{k}t}
|\tilde{k}\rangle\langle\tilde{k}|\right]  \end{split}\\
\begin{split}&  =\int_{\mathbb{R}}dt\ p(t)\operatorname{Tr}\!\left[  \mathcal{U}^\dag_{R_i}\!(H_i) e^{iG t} \mathcal{U}^\dag_{R_j}\!(H_j) \rho G e^{-iG t}\right]  \\
&  \qquad-\int_{\mathbb{R}}dt\ p(t)\operatorname{Tr}\!\left[ \mathcal{U}^\dag_{R_i}\!(H_i) G e^{iG t} \mathcal{U}^\dag_{R_j}\!(H_j) \rho e^{-iG t}\right]  \end{split}\\
&  =\operatorname{Tr}\!\left[  \Phi\!\left(\mathcal{U}^\dag_{R_i}\!(H_i)\right) \mathcal{U}^\dag_{R_j}\!(H_j) 
 G\rho \right]  -\operatorname{Tr}\!
\left[ \Phi\!\left(\mathcal{U}^\dag_{R_i}\!(H_i)\right) G \ \mathcal{U}^\dag_{R_j}\!(H_j) 
 \rho\right].
\end{align}
Putting everything together, we find that
\begin{align}
  I_{ij}^{\text{FB}}(\phi)
&  =2 \left\langle\left\{ \mathcal{U}^\dag_{R_i}\!(H_i), \mathcal{U}^\dag_{R_j}\!(H_j) \right\}\right\rangle_{\rho}\nonumber\\
&  \qquad-8\left(
\begin{array}
[c]{c}
\frac{1}{4}\left(
\begin{array}
[c]{c}
\operatorname{Re}\!\left[  \operatorname{Tr}\!\left[  \Phi\!\left(\mathcal{U}^\dag_{R_i}\!(H_i)\right) \mathcal{U}^\dag_{R_j}\!(H_j) 
 G\rho \right]
\right]  \\
-\operatorname{Re}\!\left[  \operatorname{Tr}\!
\left[ \Phi\!\left(\mathcal{U}^\dag_{R_i}\!(H_i)\right) G \ \mathcal{U}^\dag_{R_j}\!(H_j) 
 \rho\right]
\right]
\end{array}
\right)  \\
+\frac{1}{4}\left\langle \left\{  \mathcal{U}^\dag_{R_i}\!(H_i) ,\mathcal{U}^\dag_{R_j}\!(H_j) \right\}  \right\rangle _{\rho}
\end{array}
\right)  \\
\begin{split}&  =2\left\langle\left\{ \mathcal{U}^\dag_{R_i}\!(H_i), \mathcal{U}^\dag_{R_j}\!(H_j) \right\}\right\rangle_{\rho} -\operatorname{Tr}\!\left[  \Phi\!\left(\mathcal{U}^\dag_{R_i}\!(H_i)\right) \mathcal{U}^\dag_{R_j}\!(H_j) 
 G\rho \right]  \\
&  \qquad -\operatorname{Tr}\!\left[ \rho G \ \mathcal{U}^\dag_{R_j}\!(H_j) 
 \Phi\!\left(\mathcal{U}^\dag_{R_i}\!(H_i)\right) 
 \right] + \operatorname{Tr}\!
\left[ \Phi\!\left(\mathcal{U}^\dag_{R_i}\!(H_i)\right) G \ \mathcal{U}^\dag_{R_j}\!(H_j) 
 \rho\right]\\
&  \qquad + \operatorname{Tr}\!
\left[ \rho \mathcal{U}^\dag_{R_j}\!(H_j) G \ \Phi\!\left(\mathcal{U}^\dag_{R_i}\!(H_i)\right)  
 \right] 
- 2\left\langle\left\{ \mathcal{U}^\dag_{R_i}\!(H_i), \mathcal{U}^\dag_{R_j}\!(H_j) \right\}\right\rangle_{\rho}\end{split}\\
\begin{split}&  =\operatorname{Tr}\!\left[  \Phi\!\left(\mathcal{U}^\dag_{R_i}\!(H_i)\right) G \ \mathcal{U}^\dag_{R_j}\!(H_j) \rho \right]  -\operatorname{Tr}\!
\left[  \Phi\!\left(\mathcal{U}^\dag_{R_i}\!(H_i)\right) \mathcal{U}^\dag_{R_j}\!(H_j) G \rho \right]  \\
&  \qquad+\operatorname{Tr}\!\left[  \Phi\!\left(\mathcal{U}^\dag_{R_i}\!(H_i)\right) \rho \  \mathcal{U}^\dag_{R_j}\!(H_j) G \right]
-\operatorname{Tr}\!\left[ \Phi\!\left(\mathcal{U}^\dag_{R_i}\!(H_i)\right) \rho
 G \ \mathcal{U}^\dag_{R_j}\!(H_j) \right]  \end{split}\\
&  =\left\langle \left[  \left[  \mathcal{U}^\dag_{R_j}\!(H_j) , G \right]
,\Phi\!\left(\mathcal{U}^\dag_{R_i}\!(H_i)\right) \right]  \right\rangle _{\rho}.
\end{align}
This concludes the proof of Theorem~\ref{thm:FB}.

\subsection{Quantum algorithm for  estimating the Fisher--Bures information matrix elements}\label{app:FB_algo}

Let us first recall from the statement of Theorem~\ref{thm:FB} the expression for the $(i, j)$-th element of the Fisher–Bures information matrix $I_{ij}^{\operatorname{FB}}(\phi)$:
\begin{equation}
    I_{ij}^{\operatorname{FB}}(\phi) = \left\langle \left[  \left[  \mathcal{U}^\dag_{R_j}\!(H_j) , G \right],\Phi\!\left(\mathcal{U}^\dag_{R_i}\!(H_i)\right) \right]  \right\rangle _{\rho}.\label{eq:FB_app}
\end{equation}
Consider the following:
\begin{align}
    & \left\langle \left[  \left[  \mathcal{U}^\dag_{R_j}\!(H_j) , G \right],\Phi\!\left(\mathcal{U}^\dag_{R_i}\!(H_i)\right) \right]  \right\rangle _{\rho}\nonumber\\
    & = \Tr\!\bigg[ \left[  \left[  \mathcal{U}^\dag_{R_j}\!(H_j) , G \right],\Phi\!\left(\mathcal{U}^\dag_{R_i}\!(H_i)\right) \right] \rho \bigg] \\
    & = \int_\mathbb{R} dt\ p(t)\ \bigg( \Tr\!\bigg[ \left[ \left[  U^\dagger_{R_j} H_j U_{R_j} , G \right] , e^{-iGt} U^\dagger_{R_i} H_i U_{R_i} e^{iGt}\right] \rho \bigg]  \bigg).\label{eq:FB-last}
\end{align}

We are now in a position to present an algorithm (Algorithm~\ref{algo:FB-phi}) to estimate~\eqref{eq:FB_app} using its equivalent form shown in~\eqref{eq:FB-last}. At the core of our algorithm lies the quantum circuit that estimates the expected value of two nested commutators of three operators, $\frac{1}{4} \left\langle \big[ \left[U_1 , H \right], U_0\big]\right\rangle_\rho$, where $H$ is Hermitian, and $U_0$ and $U_1$ are both Hermitian and unitary (refer to Appendix B and Figure 6a of~\cite{minervini2025eqbm}). In this case, we choose $U_1 = U^\dagger_{R_j} H_j U_{R_j}$, $H = G$, and $U_0 = e^{-iGt} U^\dagger_{R_i} H_i U_{R_i} e^{iGt}$. We then make some further simplifications where possible. Specifically, the quantum circuit that plays a role in estimating the integrand of~\eqref{eq:FB-last} is depicted in Figure~\ref{fig:FB}.

\begin{remark}
\label{rem:measuring-G-phi}
In Algorithm~\ref{algo:FB-phi} below, note that we adopt a sampling approach to measuring $G$, similar to that used in \cite[Algorithm~1]{patel2024quantumboltzmannmachine}, using the Pauli decomposition $G=\sum_k \alpha_k G_k$. This seems to be necessary, as it is not obvious how to measure directly in the eigenbasis of $G$. We adopt a similar approach in other circuits that involve measuring $G$.
\end{remark}

\begin{algorithm}[H]
\caption{\texorpdfstring{$\mathtt{estimate\_FB}(i, j, \phi, \{H_\ell\}_{\ell=1}^{J}, \rho, G, p(\cdot), \varepsilon, \delta)$}{estimate first term}}
\label{algo:FB-phi}
\begin{algorithmic}[1]
\STATE \textbf{Input:} Indices $i, j \in [J]$, parameter vector $\phi = \left( \phi_{1}, \ldots,  \phi_{J}\right)^{\mathsf{T}} \in \mathbb{R}^{J}$, Hermitian operators $\{H_\ell\}_{\ell=1}^{J}$, initial state $\rho$ and $G$ is the underlying Hamiltonian, probability distribution $p(t)$ over $\mathbb{R}$, precision $\varepsilon > 0$, error probability $\delta \in (0,1)$
\STATE $N \leftarrow \lceil\sfrac{2\left\| \phi \right\|_1^2 \ln(\sfrac{2}{\delta})}{\varepsilon^2}\rceil$
\FOR{$n = 0$ to $N-1$}
\STATE Initialize the first control register to $|1\rangle\!\langle 1 |$
\STATE Initialize the second control register to $|1\rangle\!\langle 1 |$
\STATE Prepare the system register in the state $\rho$
\STATE Sample $t$ at random with probability $p(t)$ (defined in~\eqref{eq:high-peak-tent-density})
\STATE Apply the Hadamard gate and the phase gate $S$ to the first and second control registers
\STATE Apply the following unitaries to the control and system registers:
\STATE \hspace{0.6cm} \textbullet~$U_{R_i}$: gate sequence defined in~\eqref{eq:partial_gates_right} on the system register
\STATE \hspace{0.6cm} \textbullet~Controlled-$H_i$: $H_i$ is a local unitary acting on the system register, controlled by the first control register
\STATE \hspace{0.6cm} \textbullet~$U^\dagger_{R_i}$: adjoint of the gate sequence defined in~\eqref{eq:partial_gates_right} on the system register
\STATE \hspace{0.6cm} \textbullet~$e^{-iGt}$: Hamiltonian simulation for time $t$ on the system register
\STATE \hspace{0.6cm} \textbullet~$U_{R_j}$: gate sequence defined in~\eqref{eq:partial_gates_right} on the system register
\STATE \hspace{0.6cm} \textbullet~Controlled-$H_j$: $H_j$ is a local unitary acting on the system register, controlled by the second control register
\STATE \hspace{0.6cm} \textbullet~$U^\dagger_{R_j}$: adjoint of the gate sequence defined in~\eqref{eq:partial_gates_right} on the system register
\STATE Apply the Hadamard gate to the first and second control registers
\STATE Measure the two control registers in the computational basis and store the measurement outcomes~$b_n$ and $c_n$
\STATE From the Pauli decomposition of $G$ (see Remark~\ref{rem:measuring-G-phi}), $G=\sum_k \alpha_k G_k$, sample the index $k$ according to the probability distribution $|\alpha_k| / \left\| \alpha \right\|_1$, measure the system register in the eigenbasis of sign$(\alpha_k) G_k$ and store the measurement outcome $\lambda_n$
\STATE $Y_{n}^{(\operatorname{FB})} \leftarrow (-1)^{b_n}(-1)^{c_n}\lambda_n$
\ENDFOR

\STATE \textbf{return} $\overline{Y}^{(\operatorname{FB})} \leftarrow 4 \left\| \alpha \right\|_1 \times\frac{1}{N}\sum_{n=0}^{N-1}Y_{n}^{(\operatorname{FB})}$
\end{algorithmic}
\end{algorithm}

\section{Relation with Quantum Evolution Machines}

\label{app:QEM_PQC}

This appendix explains how PQCs, as defined in~\eqref{eq:param_circ}, relate to quantum evolution machines (QEMs), which are a generalization of PQCs. Our argument is related to that presented in \cite[Eqs.~(3)--(5)]{kottmann2023evaluatinganalyticgradientspulse}. Without loss of generality, we assume for the purpose of this section that the unparameterized unitaries $V_j$ in a PQC are identities; that is, for all $j\in [J]$, we assume that $V_j = I$. 

Recall that a parameterized quantum state $\rho(\phi)$ derived from a QEM $e^{-iH(\phi)}$ has the following form:
\begin{equation}
    \rho(\phi) \coloneqq e^{-iH(\phi)} \rho e^{iH(\phi)},\label{eq:qem-def}
\end{equation}
where $\rho$ is some initial mixed state and the Hamiltonian $H(\phi)$ is defined as follows:
\begin{equation}
    H(\phi) \coloneqq \sum_{j=1}^J \phi_j H_j.
\end{equation}
In the above equation, each Hamiltonian $H_j$ is defined in a similar way as we did previously for PQCs (see~\eqref{eq:U_j_def}). Next, from~\cite[Theorem 1]{minervini2025eqbm}, we have that the partial derivatives of $\rho(\phi)$  are given as follows:
\begin{equation}
    \partial_j\rho(\phi) = i[\rho(\phi), \Psi_{\phi}(H_j)]\label{eq:par-der-qem},
\end{equation}
where
\begin{equation}
    \Psi_{\phi}(X) \coloneqq \int_{0}^1 dt\ e^{-iH(\phi)t} X e^{iH(\phi)t}
\end{equation}
is a continuous random unitary channel.

Note that if we now replace QEM with a PQC to obtain a parameterized quantum state in~\eqref{eq:qem-def}, then the quantum channel $\Psi_{\phi}$ in the expression of the partial derivatives in~\eqref{eq:par-der-qem} gets replaced by a simpler channel, which is a unitary channel $\mathcal{U}_{L_{j+1}}$ as defined in~\eqref{eq:q_channel_left} (see Lemma~\ref{lem:part_der} for the complete derivation). In other words, we have the following:
\begin{equation}
    \partial_j\rho(\phi) = i[\rho(\phi), \mathcal{U}_{L_{j+1}}(H_j)].
\end{equation}
In essence, when considering PQCs, we obtain a discrete channel $\mathcal{U}_{L_{j+1}}$ instead of a continuous channel like $\Psi_{\phi}$, and in this section, our aim is to investigate at what point this transition from discrete to continuous channel occurs. 

To illustrate this, consider the following:
\begin{align}
    \partial_j\rho(\phi) & = \partial_j\!\left(e^{-iH(\phi)} \rho e^{iH(\phi)} \right)\\
    & = \partial_j\!\left( \lim_{r\rightarrow \infty} \left(e^{-iH(\phi)/r} \right)^r \rho \left(e^{iH(\phi)/r} \right)^r \right).
\end{align}
Next, using the first-order Trotter formula~\cite{Suzuki1976}, we transform the unitary $e^{-iH(\phi)/r}$ into a product of simpler unitaries~$e^{-i\phi_j H_j/r}$, resulting in a PQC $U(\phi/r)$, which we denote as $U$ for simplicity, hiding the explicit dependence on $\phi$ and $r$. To this end, we have
\begin{align}
    \partial_j\rho(\phi) & = \partial_j\!\Bigg( \lim_{r\rightarrow \infty} \Bigg(\underbrace{\prod_{j_1=1}^{J} e^{-i \phi_{j_1} H/r}}_{= U}\Bigg)^r \rho \Bigg(\underbrace{\prod_{j_2=J}^{1} e^{i \phi_{j_2} H/r}}_{= U^{\dagger}}\Bigg)^r \Bigg)\\
    & = \partial_j\!\left( \lim_{r\rightarrow \infty} U^r \rho  U^{r\dagger}\right)\\
    & = \lim_{r\rightarrow \infty} \partial_j\!\left(U^{r} \rho  U^{r\dagger}\right)\\
    & = \lim_{r\rightarrow \infty} \left(\partial_j U^{r} \right) \rho  U^{r\dagger} + U^{r} \rho  \left(\partial_j U^{r\dagger}\right)\\
    & = \lim_{r\rightarrow \infty} \left( \sum_{k=1}^r U^{k-1} (\partial_j U) U^{r-k} \right) \rho  U^{r\dagger} + U^{r} \rho  \left( \sum_{k=1}^r U^{k-1\dagger} (\partial_j U^{\dagger}) U^{r-k \dagger} \right)\\
    & = \lim_{r\rightarrow \infty} \left( \sum_{k=1}^r U^{k-1} \left(U_{L_{j+1}} \frac{-i H_j}{r} U_{R_j}\right) U^{r-k} \right) \rho  U^{r\dagger} + U^{r} \rho  \left( \sum_{k=1}^r U^{k-1\dagger} \left( U^\dagger_{R_j} \frac{iH_j}{r} U^\dagger_{L_{j+1}}\right) U^{r-k \dagger} \right)\\
    & = \lim_{r\rightarrow \infty} \frac{-i}{r} \left( \left( \sum_{k=1}^r U^{k-1} \left(U_{L_{j+1}} H_j U_{R_j}\right) U^{r-k} \right) \rho  U^{r\dagger} - U^{r} \rho  \left( \sum_{k=1}^r U^{k-1\dagger} \left( U^\dagger_{R_j} H_j U^\dagger_{L_{j+1}}\right) U^{r-k \dagger} \right)\right)\\
    & = \lim_{r\rightarrow \infty} \frac{-i}{r} \Bigg( \left( \sum_{k=1}^r U^{k-1} \left(U_{L_{j+1}} H_j U_{L_{j+1}}^{\dagger} U_{L_{j+1}} U_{R_j}\right) U^{r-k} \right) \rho  U^{r\dagger} \notag\\
    & \qquad \qquad \qquad - U^{r} \rho  \left( \sum_{k=1}^r U^{k-1\dagger} \left( U^\dagger_{R_j} U_{L_{j+1}}^{\dagger} U_{L_{j+1}} H_j U^\dagger_{L_{j+1}}\right) U^{r-k \dagger} \right)\Bigg)\\
    & = \lim_{r\rightarrow \infty} \frac{-i}{r} \Bigg( \left( \sum_{k=1}^r U^{k-1} \left(U_{L_{j+1}} H_j U_{L_{j+1}}^{\dagger} U\right) U^{r-k} \right) \rho  U^{r\dagger} \notag\\
    & \qquad \qquad \qquad - U^{r} \rho  \left( \sum_{k=1}^r U^{k-1\dagger} \left( U^{\dagger} U_{L_{j+1}} H_j U^\dagger_{L_{j+1}}\right) U^{r-k \dagger} \right)\Bigg)\\
    & = \lim_{r\rightarrow \infty} \frac{-i}{r} \Bigg( \left( \sum_{k=1}^r U^{k-1} \mathcal{U}_{L_{j+1}}(H_j)  U^{r-k+1} \right) \rho  U^{r\dagger} - U^{r} \rho  \left( \sum_{k=1}^r U^{k\dagger} \mathcal{U}_{L_{j+1}} (H_j)  U^{r-k \dagger} \right)\Bigg)\\
    & = \lim_{r\rightarrow \infty} \frac{-i}{r} \Bigg( \left( \sum_{k=1}^r U^{k-1} \mathcal{U}_{L_{j+1}}(H_j) U^{k-1 \dagger}U^{k-1}  U^{r-k+1} \right) \rho  U^{r\dagger} - U^{r} \rho  \left( \sum_{k=1}^r U^{k\dagger} U^{r-k \dagger} U^{r-k}\mathcal{U}_{L_{j+1}} (H_j)  U^{r-k \dagger} \right)\Bigg)\\
    & = \lim_{r\rightarrow \infty} \frac{-i}{r} \Bigg( \left( \sum_{k=1}^r U^{k-1} \mathcal{U}_{L_{j+1}}(H_j) U^{k-1 \dagger}U^{r} \right) \rho  U^{r\dagger} - U^{r} \rho  \left( \sum_{k=1}^r U^{r\dagger} U^{r-k}\mathcal{U}_{L_{j+1}} (H_j)  U^{r-k \dagger} \right)\Bigg)\\
    & = \lim_{r\rightarrow \infty} \frac{-i}{r} \Bigg( \left( \sum_{k=1}^r U^{k-1} \mathcal{U}_{L_{j+1}}(H_j) U^{k-1 \dagger}\right) U^{r} \rho  U^{r\dagger} - U^{r} \rho U^{r\dagger} \left( \sum_{k=1}^r U^{r-k}\mathcal{U}_{L_{j+1}} (H_j)  U^{r-k \dagger} \right)\Bigg)\\
    & = \lim_{r\rightarrow \infty} \frac{-i}{r} \Bigg( \left( \sum_{k=1}^r U^{k-1} \mathcal{U}_{L_{j+1}}(H_j) U^{k-1 \dagger}\right) U^{r} \rho  U^{r\dagger} - U^{r} \rho U^{r\dagger} \left( \sum_{k=1}^r U^{k-1}\mathcal{U}_{L_{j+1}} (H_j)  U^{k-1 \dagger} \right)\Bigg)\\
    & = \lim_{r\rightarrow \infty} \frac{-i}{r} \Bigg[ \left( \sum_{k=1}^r U^{k-1} \mathcal{U}_{L_{j+1}}(H_j) U^{k-1 \dagger}\right),  U^{r} \rho  U^{r\dagger} \Bigg]\\
    & = \lim_{r\rightarrow \infty} \frac{-i}{r} \Bigg[ \left( \sum_{k=1}^r \left(e^{-iH(\phi)/r} + O\!\left(\frac{1}{r^2}\right)\right)^{k-1} \mathcal{U}_{L_{j+1}}(H_j) \left(e^{iH(\phi)/r} + O\!\left(\frac{1}{r^2}\right)\right)^{k-1 \dagger}\right),  U^{r} \rho  U^{r\dagger} \Bigg] \\
    & = \lim_{r\rightarrow \infty} -i \Bigg[ \left( \sum_{k=1}^r \frac{1}{r} \left(e^{-iH(\phi)/r} + O\!\left(\frac{1}{r^2}\right)\right)^{k-1} \mathcal{U}_{L_{j+1}}(H_j) \left(e^{iH(\phi)/r} + O\!\left(\frac{1}{r^2}\right)\right)^{k-1 \dagger}\right),  U^{r} \rho  U^{r\dagger} \Bigg].
\end{align}
The exchange of the limit $\lim_{r \to \infty}$ and the derivative $\partial_j$ is possible because, for all $r \in \mathbb{N}$, the derivative $\partial_j U^r$ converges uniformly.
Now, taking the limit $r \rightarrow \infty$, the sum becomes an integral, and we get
\begin{align}
    \partial_j\rho(\phi) & = -i \Bigg[ \int_{0}^1 dt\ e^{-iH(\phi)t} H_j e^{iH(\phi)t} ,  e^{-iH(\phi)} \rho  e^{iH(\phi)} \Bigg]\\
    & = -i \left[ \Psi_{\phi}(H_j),  
    \rho(\phi) \right]\\
    & = i \left[\rho(\phi), \Psi_{\phi}(H_j)  \right],
\end{align}
where we used the fact that $\lim_{r\rightarrow \infty}\mathcal{U}_{L_{j+1}}(H_j) = H_j $ and note the implicit dependence of $\mathcal{U}_{L_{j+1}}$ on $r$.

\section{Wigner--Yanese information matrix elements}\label{app:WY}

\subsection{Proof of Lemma~\ref{thm:WY-partial-derivatives}}

\label{app:gradient_purified_state}

Using the notations in the statement of Lemma~\ref{thm:WY-partial-derivatives}, consider that
\begin{align}
    \partial_j \sqrt{\rho(\phi)} & = \partial_j \left( U(\phi) \sqrt{\rho} U^\dagger (\phi) \right)\\
    & = (\partial_j U(\phi)) \sqrt{\rho} U^\dagger (\phi) + U(\phi) \sqrt{\rho} (\partial_j U^\dagger (\phi)) \\
    & = i U_{L_{j+1}} H_j  U_{R_{j}} \sqrt{\rho} U^\dagger (\phi) + i U(\phi) \sqrt{\rho} U^\dagger_{R_j} H_j U^\dagger_{L_{j+1}} \\
    & = -i U_{L_{j+1}} H_j U^\dagger_{L_{j+1}} U_{L_{j+1}} U_{R_{j}} \sqrt{\rho} U^\dagger (\phi) + i U(\phi) \sqrt{\rho} U^\dagger_{R_j} U^\dagger_{L_{j+1}} U_{L_{j+1}} H_j U^\dagger_{L_{j+1}}\\
    & = -i U_{L_{j+1}} H_j U^\dagger_{L_{j+1}} U(\phi) \sqrt{\rho} U^\dagger (\phi) + i U(\phi) \sqrt{\rho} U^\dagger(\phi) U_{L_{j+1}} H_j U^\dagger_{L_{j+1}}\\
    & = -i U_{L_{j+1}} H_j U^\dagger_{L_{j+1}} \sqrt{\rho(\phi)} + i \sqrt{\rho(\phi)} U_{L_{j+1}} H_j U^\dagger_{L_{j+1}}\\
    & = i \left[ \sqrt{\rho(\phi)}, U_{L_{j+1}} H_j U^\dagger_{L_{j+1}} \right]\\
    & = i \left[ \sqrt{\rho(\phi)}, \mathcal{U}_{L_{j+1}}\!(H_j)\right].
\end{align}
So, we have that
\begin{align}
    |\partial_{j}\psi(\phi)\rangle 
    & = \frac{\partial}{\partial \phi_j}\left(  \sqrt{\rho(\phi)}\otimes I\right)  |\Gamma\rangle\\
    & = i \left( \left[ \sqrt{\rho(\phi)}, \mathcal{U}_{L_{j+1}}\!(H_j) \right] \otimes I\right)  |\Gamma\rangle.
\end{align}
This concludes the proof of Lemma~\ref{thm:WY-partial-derivatives}.

\subsection{Proof of Theorem~\ref{thm:WY}}\label{app:WY-proof}
The derivation follows a similar approach to that of \cite[Theorem~15]{minervini2025eqbm}, but with some distinctions. Therefore, we provide a detailed proof below. Let $\partial_i \equiv \frac{\partial}{\partial \phi_i}$ and $\partial_j \equiv \frac{\partial}{\partial \phi_j}$ denote shorthand notations for partial derivatives. Using Proposition~\ref{prop:FB-WY-canonical-purifications}, the Wigner--Yanase
information of a parameterized family $\left(  \sigma(\gamma)\right)
_{\gamma\in\mathbb{R}^{L}}$ is equal to the Fisher--Bures information of the corresponding pure
parameterized fa\-mi\-ly $\left(  \varphi^{\sigma}(\gamma)\right)  _{\gamma
\in\mathbb{R}^{L}}$, where $\varphi^{\sigma}(\gamma)$ denotes a canonical
purification of $\sigma(\gamma)$ defined in~\eqref{eq:canonical-purification-FB-WY}. Recalling that the Fisher--Bures information of pure
states is given by \cite[Eq.~(132)]{Sidhu2020}
\begin{equation}
    I^{\operatorname{FB}}_{ij}(\gamma) =4\operatorname{Re}\!\left[  \left\langle \partial_{i}\psi(\gamma)|\partial_{j}\psi(\gamma)\right\rangle -\left\langle \partial_{i}\psi (\gamma)|\psi(\gamma)\right\rangle \left\langle \psi(\gamma)|\partial_{j} \psi(\gamma)\right\rangle \right]  ,\label{eq:formula-FB-pure}
\end{equation}
we proceed to evaluate each term in this expression.
Using Lemma~\ref{thm:WY-partial-derivatives}, we find that
\begin{align}
    \langle \partial_i \psi(\phi)| \partial_j \psi(\phi) \rangle 
    & =  -\langle \Gamma | \left( \left[ \sqrt{\rho(\phi)}, \mathcal{U}_{L_{i+1}}\!(H_i) \right] \otimes I\right) \left( \left[ \sqrt{\rho(\phi)}, \mathcal{U}_{L_{j+1}}\!(H_j) \right] \otimes I\right) | \Gamma \rangle\\
    & =  -\langle \Gamma | \left[ \sqrt{\rho(\phi)}, \mathcal{U}_{L_{i+1}}\!(H_i) \right] \left[ \sqrt{\rho(\phi)}, \mathcal{U}_{L_{j+1}}\!(H_j) \right] \otimes I | \Gamma \rangle \\
    \begin{split}& =  -\langle \Gamma | \sqrt{\rho(\phi)} \mathcal{U}_{L_{i+1}}\!(H_i) \sqrt{\rho(\phi)}  \mathcal{U}_{L_{j+1}}\!(H_j) \otimes I | \Gamma \rangle \\
    & \hspace{0.4cm} + \langle \Gamma | \mathcal{U}_{L_{i+1}}\!(H_i) \sqrt{\rho(\phi)}  \sqrt{\rho(\phi)} \mathcal{U}_{L_{j+1}}\!(H_j) \otimes I | \Gamma \rangle\\
    & \hspace{0.4cm} +  \langle \Gamma | \sqrt{\rho(\phi)}  \mathcal{U}_{L_{i+1}}\!(H_i) \mathcal{U}_{L_{j+1}}\!(H_j) \sqrt{\rho(\phi)} \otimes I | \Gamma \rangle\\
    & \hspace{0.4cm} -  \langle \Gamma | \mathcal{U}_{L_{i+1}}\!(H_i) \sqrt{\rho(\phi)}  \mathcal{U}_{L_{j+1}}\!(H_j) \sqrt{\rho(\phi)} \otimes I | \Gamma \rangle \end{split}\\
    \begin{split}& = - \Tr\!\left[ \sqrt{\rho(\phi)} \mathcal{U}_{L_{i+1}}\!(H_i) \sqrt{\rho(\phi)}  \mathcal{U}_{L_{j+1}}\!(H_j) \right] \\
    & \hspace{0.4cm} +  \Tr\!\left[ \mathcal{U}_{L_{i+1}}\!(H_i) \sqrt{\rho(\phi)}  \sqrt{\rho(\phi)} \mathcal{U}_{L_{j+1}}\!(H_j) \right]\\
    & \hspace{0.4cm} + \Tr\!\left[ \sqrt{\rho(\phi)} \mathcal{U}_{L_{i+1}}\!(H_i) \mathcal{U}_{L_{j+1}}\!(H_j) \sqrt{\rho(\phi)} \right] \\
    & \hspace{0.4cm} -  \Tr\!\left[ \mathcal{U}_{L_{i+1}}\!(H_i) \sqrt{\rho(\phi)}  \mathcal{U}_{L_{j+1}}\!(H_j) \sqrt{\rho(\phi)} \right] \end{split}\\
    \begin{split}& = - 2 \Tr\!\left[  \mathcal{U}_{L_{i+1}}\!(H_i) \sqrt{\rho(\phi)}  \mathcal{U}_{L_{j+1}}\!(H_j) \sqrt{\rho(\phi)} \right] \\
    & \hspace{0.4cm} +  \Tr\!\left[ \left\{ \mathcal{U}_{L_{i+1}}\!(H_i) , \mathcal{U}_{L_{j+1}}\!(H_j) \right\} \rho(\phi)  \right]\end{split}\\
    \begin{split}& = - 2 \Tr\!\left[  U_{L_{i+1}} H_i U^\dagger_{L_{i+1}} U\!(\phi) \sqrt{\rho} U^\dagger\!(\phi)  U_{L_{j+1}} H_j U^\dagger_{L_{j+1}} U\!(\phi) \sqrt{\rho} U^\dagger\!(\phi) \right] \\
    & \hspace{0.4cm} +  \Tr\!\left[ \left\{ U_{L_{i+1}} H_i U^\dagger_{L_{i+1}} , U_{L_{j+1}} H_j U^\dagger_{L_{j+1}} \right\} U\!(\phi) \rho U^\dagger\!(\phi)  \right]\end{split}\\
    \begin{split}& = - 2 \Tr\!\left[  U^\dagger_{R_{i}} H_i U_{R_{i}} \sqrt{\rho} \ U^\dagger_{R_{j}} H_j U_{R_{j}} \sqrt{\rho} \right] \\
    & \hspace{0.4cm} +  \Tr\!\left[ \left\{ U^\dagger_{R_{i}} H_i U_{R_{i}} , U^\dagger_{R_{j}} H_j U_{R_{j}} \right\} \rho  \right]\end{split} \label{pass:comm_rel}\\
    \begin{split}& = - 2 \Tr\!\left[  \mathcal{U}^\dag_{R_i}\!(H_i) \sqrt{\rho} \ \mathcal{U}^\dag_{R_j}\!(H_j) \sqrt{\rho} \right] \\
    & \hspace{0.4cm} + \left\langle \left\{ \mathcal{U}^\dag_{R_i}\!(H_i) , \mathcal{U}^\dag_{R_j}\!(H_j) \right\}\right\rangle_\rho,\end{split}
\end{align}
Additionally, consider that
\begin{align}
    \langle \psi(\phi)| \partial_j \psi(\phi) \rangle 
    & = i \langle \Gamma | \left( \sqrt{\rho(\phi)} \otimes I\right) \left( \left[ \sqrt{\rho(\phi)}, \mathcal{U}_{L_{j+1}}\!(H_j) \right] \otimes I\right) | \Gamma \rangle\\
    & = i \langle \Gamma | \sqrt{\rho(\phi)} \left[ \sqrt{\rho(\phi)}, \mathcal{U}_{L_{j+1}}\!(H_j) \right] \otimes I | \Gamma \rangle\\
    & = i \Tr\!\left[ \mathcal{U}_{L_{j+1}}\!(H_j) \rho(\phi) \right] - i \Tr\!\left[ \mathcal{U}_{L_{j+1}}\!(H_j) \rho(\phi) \right]\\
    & = 0. 
\end{align}
Finally, substituting these results into~\eqref{eq:formula-FB-pure}, we obtain the Wigner–Yanase information matrix elements: 
\begin{align}
    I^{\operatorname{WY}}_{ij}(\phi) & = 4 \operatorname{Re}\!\left[  \langle \partial_i \psi(\phi)| \partial_j \psi(\phi) \rangle - \langle \partial_i \psi(\phi)|  \psi(\phi) \rangle \langle \psi(\phi)| \partial_j \psi(\phi) \rangle\right]\\
    & = 4 \operatorname{Re}\!\left[ -2 \Tr\!\left[ \mathcal{U}^\dag_{R_i}\!(H_i) \sqrt{\rho} \ \mathcal{U}^\dag_{R_j}\!(H_j) \sqrt{\rho} \right] +  \left\langle \left\{ \mathcal{U}^\dag_{R_i}\!(H_i) , \mathcal{U}^\dag_{R_j}\!(H_j) \right\}\right\rangle_\rho \right]\\
    & = -8 \Tr\!\left[  \mathcal{U}^\dag_{R_i}\!(H_i) \sqrt{\rho} \ \mathcal{U}^\dag_{R_j}\!(H_j) \sqrt{\rho} \right] + 4 \left\langle \left\{ \mathcal{U}^\dag_{R_i}\!(H_i) , \mathcal{U}^\dag_{R_j}\!(H_j) \right\}\right\rangle_\rho .
\end{align}
This concludes the proof of Theorem~\ref{thm:WY}.

\subsection{Quantum algorithms for estimating the Wigner--Yanase information matrix elements}\label{app:WY_algo}

Let us recall from the statement of Theorem~\ref{thm:WY} the expression for the $(i, j)$-th element of the Wigner--Yanase information matrix $I_{ij}^{\operatorname{WY}}(\phi)$:
\begin{equation}
    I_{ij}^{\operatorname{WY}}(\phi) = 4 \left\langle \left\{ \mathcal{U}^\dag_{R_i}\!(H_i) , \mathcal{U}^\dag_{R_j}\!(H_j) \right\}\right\rangle_\rho - 8 \Tr\!\left[  \mathcal{U}^\dag_{R_i}\!(H_i) \sqrt{\rho} \ \mathcal{U}^\dag_{R_j}\!(H_j) \sqrt{\rho} \right].\label{eq:WY_app}
\end{equation}
We present how to estimate each term of the above expression separately in the following subsections.

\subsubsection{Estimation of the first term}\label{app:WY-first_term}
We begin by outlining the estimation of the first term in~\eqref{eq:WY_app}.  Consider the following:
\begin{align}
    4 \left\langle \left\{ \mathcal{U}^\dag_{R_i}\!(H_i) , \mathcal{U}^\dag_{R_j}\!(H_j) \right\}\right\rangle_\rho
    & = 4\Tr\!\left[ \left\{ \mathcal{U}^\dag_{R_i}\!(H_i) , \mathcal{U}^\dag_{R_j}\!(H_j) \right\}  \rho \right]\\
    & = 4\operatorname{Tr}\!\left[ \left\{
     U^\dagger_{R_i}  H_i U_{R_i} , U^\dagger_{R_j} H_j U^\dagger_{R_j} \right\} \rho \right] .\label{eq:WY_app-alt}
\end{align}
We are now in a position to present an algorithm to estimate the second term of~\eqref{eq:WY_app} using its equivalent form shown in~\eqref{eq:WY_app-alt}. This term closely resembles the second term in the expression of \cite[Theorem~4]{minervini2025eqbm}, and thus the estimation algorithm is analogous to the one presented in \cite[Appendix I.2]{minervini2025eqbm}. So here we provide just a high-level description of the algorithm.  At its core, the algorithm relies on a quantum circuit that estimates the expected value of the anticommutator of two operators (see \cite[Appendix B]{minervini2025eqbm}). The output of this circuit is $\frac{1}{2} \left\langle \left\{H,U\right\} \right\rangle_{\rho}$, where $H$ is Hermitian and $U$ is Hermitian and unitary. In this case, we choose $H = U^\dagger_{R_i}  H_i U_{R_i}$, and $U = U^\dagger_{R_j} H_j U^\dagger_{R_j}$. Accordingly, the quantum circuit that plays a role in estimating the quantity in~\eqref{eq:WY_app-alt} is depicted in Figure~\ref{fig:WY-1}. The algorithm involves running this circuit $N$ times, where $N$ is determined by the desired precision and error probability. The final estimation of the first term of~\eqref{eq:WY_app} is obtained by averaging the outputs of the $N$ runs and multiplying the result by 8.

\subsubsection{Estimation of the second term}\label{app:WY-second_term}

In order to estimate the second term of~\eqref{eq:WY_app}, we assume that one  has access to the canonical purification
$|\psi(\phi)\rangle$ of a parameterized quantum state, defined as
\begin{equation}
    |\psi(\phi)\rangle \coloneqq \left(  \sqrt{\rho(\phi)}\otimes
I\right)  |\Gamma\rangle.
\label{eq:canon-pure-QBM}
\end{equation}
Under this assumption, the following identity implies that one can estimate the second term of~\eqref{eq:WY_app} efficiently:
\begin{equation}
\operatorname{Tr}\!\left[ \mathcal{U}^\dagger_{R_i}(H_i) \sqrt{\rho(\phi)} \ \mathcal{U}^\dagger_{R_j}(H_j)  \sqrt{\rho(\phi)}\right] = 
\langle\psi(\phi)|\left(  \mathcal{U}^\dagger_{R_i}(H_i) \otimes\left[
\mathcal{U}^\dagger_{R_j}(H_j) \right]  ^{T}\right)  |\psi(\phi)\rangle.
\label{eq:identity-WY-canon-pur}
\end{equation}
The identity in~\eqref{eq:identity-WY-canon-pur} follows because
\begin{align}
&  \langle\psi(\phi)|\left(  \mathcal{U}^\dagger_{R_i}(H_i) \otimes\left[
\mathcal{U}^\dagger_{R_j}(H_j) \right]  ^{T}\right)  |\psi(\phi)\rangle\nonumber\\
&  =\langle\Gamma|\left(  \sqrt{\rho(\phi)} \ \mathcal{U}^\dagger_{R_i}(H_i) \sqrt{\rho(\phi)}\otimes\left[  \mathcal{U}^\dagger_{R_j}(H_j)\right]  ^{T}\right)  |\Gamma\rangle\\
&  =\langle\Gamma|\left(  \sqrt{\rho(\phi)} \ \mathcal{U}^\dagger_{R_i}(H_i) \sqrt{\rho(\phi)} \ \mathcal{U}^\dagger_{R_j}(H_j) \otimes I \right)
|\Gamma\rangle\\
&  =\operatorname{Tr}\!\left[  \sqrt{\rho(\phi)} \ \mathcal{U}^\dagger_{R_i}(H_i) \sqrt{\rho(\phi)} \ \mathcal{U}^\dagger_{R_j}(H_j) \right]  .
\end{align}
The second equality follows from the transpose trick \cite[Exercise~3.7.12]{Wbook17}.
Thus, in order to estimate the right-hand side of~\eqref{eq:identity-WY-canon-pur}, we need to be able to measure the expectation of the operator $\left[  \mathcal{U}^\dagger_{R_j}(H_j) \right]  ^{T}$. Consider that
\begin{align}
    \left[  \mathcal{U}^\dagger_{R_j}(H_j) \right]  ^{T} 
& = \left[ U^\dagger_{R_j} H_j U_{R_j} \right]^{T}\\
& = U_{R_j}^{T} H^{T} _j \left(U^\dagger_{R_j}\right)^{T} \\
& = U_{R_j}^{T} \overline{H}_j \  \overline{U}_{R_j}.\label{step:transpose}
\end{align}
In~\eqref{step:transpose}, we used the fact that the transpose of a Hermitian matrix is equal to its complex conjugate and the transpose of the adjoint of a matrix is equal to its complex conjugate.
Then, adopting the shorthand $\psi(\phi)\equiv|\psi(\phi)\rangle\!
\langle\psi(\phi)|$, applying the definition of $\mathcal{U}^\dagger_{R_j}$ in~\eqref{eq:q_channel_right}
and cyclicity of trace, consider that
\begin{align}
& \langle\psi(\phi)|\left(  \mathcal{U}^\dagger_{R_i}(H_i) \otimes\left[
\mathcal{U}^\dagger_{R_j}(H_j) \right]  ^{T}\right)  |\psi(\phi)\rangle \notag \\
& =\operatorname{Tr}\!\left[  \left(  \mathcal{U}^\dagger_{R_i}(H_i) \otimes\left[
\mathcal{U}^\dagger_{R_j}(H_j) \right]  ^{T}\right)
\psi(\phi)\right]  \\
& = \operatorname{Tr}\!\left[  \left(
H_{i}\otimes \overline{H}_{j} \right)  \mathcal{V}(\psi
(\phi))\right] ,
\end{align}
where $\mathcal{V}$ is the following unitary channel:
\begin{equation}
    \mathcal{V}(Y)\coloneqq \left( U_{R_i} \otimes \overline{U}_{R_j} \right)  Y \left( U^\dagger_{R_i} \otimes U_{R_j}^{T} \right) .
\end{equation}

Thus, a quantum algorithm for estimating the second term of~\eqref{eq:WY_app} consists of
repeating the following steps and averaging: prepare the canonical
purification $\psi(\phi)$ in~\eqref{eq:canon-pure-QBM}, apply the unitary channel $\mathcal{V}$ to $\psi(\phi)$, and measure the observable
$H_{i}\otimes \overline{H}_{j}$. The final estimation is obtained by multiplying the result by $-8$. The respective quantum circuit is shown in Figure~\ref{fig:WY-1}.

\section{Kubo--Mori information matrix elements}\label{app:KM}

\subsection{Proof of Theorem~\ref{thm:KM}}\label{app:KM-proof}
The derivation follows a similar approach to that of \cite[Theorem~19]{minervini2025eqbm}, but with some distinctions. Therefore, we provide a detailed proof below. Using the result from Lemma~\ref{lem:matr_part_der}, consider that
\begin{align}
    I^{\operatorname{KM}}_{ij}(\phi) 
    & = \sum_{k,\ell}\frac{\ln\lambda_{k}-\ln\lambda_{\ell}}{\lambda_{k} - \lambda_{\ell}} \langle k|\partial_{i}\rho(\phi)|\ell\rangle\langle \ell|\partial_{j}\rho(\phi)|k\rangle \\
    & = \sum_{k,\ell}\frac{\ln\lambda_{k}-\ln\lambda_{\ell}}{\lambda_{k} -\lambda_{\ell}} \left( i(\lambda_k - \lambda_\ell) \langle \tilde{k}| \mathcal{U}^\dag_{R_i}\!(H_i) |\tilde{\ell}\rangle \right) \left( i(\lambda_\ell - \lambda_k) \langle\tilde{\ell}| \mathcal{U}^\dag_{R_j}\!(H_j) |\tilde{k}\rangle \right)\\
    & = \sum_{k,\ell}\frac{\ln\lambda_{k}-\ln\lambda_{\ell}}{\lambda_{k} -\lambda_{\ell}} (\lambda_k - \lambda_\ell)^2 \langle \tilde{k}| \mathcal{U}^\dag_{R_i}\!(H_i) |\tilde{\ell}\rangle  \langle\tilde{\ell}| \mathcal{U}^\dag_{R_j}\!(H_j) |\tilde{k}\rangle \\
    & = \sum_{k,\ell} (\ln\lambda_{k}-\ln\lambda_{\ell}) (\lambda_k - \lambda_\ell) \langle \tilde{k}| \mathcal{U}^\dag_{R_i}\!(H_i) |\tilde{\ell}\rangle  \langle\tilde{\ell}| \mathcal{U}^\dag_{R_j}\!(H_j) |\tilde{k}\rangle \\
    \begin{split}& = \sum_{k,\ell}\lambda_{k}\ln\lambda_{k}\langle \tilde{k}| \mathcal{U}^\dag_{R_i}\!(H_i) |\tilde{\ell}\rangle  \langle\tilde{\ell}| \mathcal{U}^\dag_{R_j}\!(H_j) |\tilde{k}\rangle\\
    &  \qquad-\sum_{k,\ell}\lambda_{\ell}\ln\lambda_{k}\langle \tilde{k}| \mathcal{U}^\dag_{R_i}\!(H_i) |\tilde{\ell}\rangle  \langle\tilde{\ell}| \mathcal{U}^\dag_{R_j}\!(H_j) |\tilde{k}\rangle\\
    &  \qquad-\sum_{k,\ell}\lambda_{k}\ln\lambda_{\ell}\langle \tilde{k}| \mathcal{U}^\dag_{R_i}\!(H_i) |\tilde{\ell}\rangle  \langle\tilde{\ell}| \mathcal{U}^\dag_{R_j}\!(H_j) |\tilde{k}\rangle\\
    &  \qquad+\sum_{k,\ell}\lambda_{\ell}\ln\lambda_{\ell}\langle \tilde{k}| \mathcal{U}^\dag_{R_i}\!(H_i) |\tilde{\ell}\rangle  \langle\tilde{\ell}| \mathcal{U}^\dag_{R_j}\!(H_j) |\tilde{k}\rangle\end{split}\\
    \begin{split}& = \operatorname{Tr}\!\left[ \mathcal{U}^\dag_{R_j}\!(H_j)  \rho \ln\rho \ \mathcal{U}^\dag_{R_i}\!(H_i)\right] - 
    \operatorname{Tr}\!\left[ \ln\rho \ \mathcal{U}^\dag_{R_i}\!(H_i) \rho \   \mathcal{U}^\dag_{R_j}\!(H_j)\right] \\
    &  \qquad - \operatorname{Tr}\!\left[ \rho \ \mathcal{U}^\dag_{R_i}\!(H_i) \ln\rho \  \mathcal{U}^\dag_{R_j}\!(H_j) \right]
    + \operatorname{Tr}\!\left[  \mathcal{U}^\dag_{R_i}\!(H_i) \rho \ln\rho \ \mathcal{U}^\dag_{R_j}\!(H_j) \right]\end{split}\\
    \begin{split}& = - \operatorname{Tr}\!\left[ \mathcal{U}^\dag_{R_j}\!(H_j)  \rho  G \ \mathcal{U}^\dag_{R_i}\!(H_i)\right] + 
    \operatorname{Tr}\!\left[ G \ \mathcal{U}^\dag_{R_i}\!(H_i) \rho \  \mathcal{U}^\dag_{R_j}\!(H_j)\right] \\
    &  \qquad + \operatorname{Tr}\!\left
    [ \rho \ \mathcal{U}^\dag_{R_i}\!(H_i) G \  \mathcal{U}^\dag_{R_j}\!(H_j)\right]
    - \operatorname{Tr}\!\left[  \mathcal{U}^\dag_{R_i}\!(H_i) \rho  G \ \mathcal{U}^\dag_{R_j}\!(H_j) \right]\end{split}\\
    \begin{split}& = - \operatorname{Tr}\!\left[ \mathcal{U}^\dag_{R_i}\!(H_i) \mathcal{U}^\dag_{R_j}\!(H_j) G  \rho\right] + 
    \operatorname{Tr}\!\left[ \mathcal{U}^\dag_{R_j}\!(H_j) G \ \mathcal{U}^\dag_{R_i}\!(H_i) \rho\right] \\
    &  \qquad + \operatorname{Tr}\!\left[ \mathcal{U}^\dag_{R_i}\!(H_i) G \  \mathcal{U}^\dag_{R_j}\!(H_j) \rho \right]
    - \operatorname{Tr}\!\left[ G \ \mathcal{U}^\dag_{R_j}\!(H_j)  \mathcal{U}^\dag_{R_i}\!(H_i) \rho \right]\end{split}\\
    \begin{split}& = \operatorname{Tr}\!\left[ \mathcal{U}^\dag_{R_j}\!(H_j) G \ \mathcal{U}^\dag_{R_i}\!(H_i) \rho\right] 
    - \operatorname{Tr}\!\left[ G \ \mathcal{U}^\dag_{R_j}\!(H_j) \mathcal{U}^\dag_{R_i}\!(H_i) \rho \right] \\
    &  \qquad - \operatorname{Tr}\!\left[ \mathcal{U}^\dag_{R_i}\!(H_i) \mathcal{U}^\dag_{R_j}\!(H_j) G \rho\right] +
    \operatorname{Tr}\!\left[ \mathcal{U}^\dag_{R_i}\!(H_i) G \ \mathcal{U}^\dag_{R_j}\!(H_j) \rho \right] \end{split}\\  
    & = \operatorname{Tr}\!\bigg[ \left[ \left[ \mathcal{U}^\dag_{R_j}\!(H_j) , G \right] , \mathcal{U}^\dag_{R_i}\!(H_i) \right] \rho \bigg] \\
    & = \left\langle \left[ \left[ \mathcal{U}^\dag_{R_j}\!(H_j) , G \right] \! , \mathcal{U}^\dag_{R_i}\!(H_i) \right] \right\rangle_\rho.
\end{align}
This concludes the proof of Theorem~\ref{thm:KM}.

\subsection{Quantum algorithm for  estimating the Kubo--Mori information matrix elements}\label{app:KM_algo}

Let us first recall from the statement of Theorem~\ref{thm:KM} the expression for the $(i, j)$-th element of the Kubo--Mori information matrix $I_{ij}^{\operatorname{KM}}(\phi)$:
\begin{equation}
    I^{\operatorname{KM}}_{ij}(\phi) = \left\langle \left[ \left[ \mathcal{U}^\dag_{R_j}\!(H_j) , G \right] \!, \mathcal{U}^\dag_{R_i}\!(H_i) \right] \right\rangle_\rho.\label{eq:KM_app}
\end{equation}
Consider the following:
\begin{align}
    \left\langle \left[ \left[ \mathcal{U}^\dag_{R_j}\!(H_j) , G \right] \!, \mathcal{U}^\dag_{R_i}\!(H_i) \right] \right\rangle_\rho 
    & = \Tr\! \bigg[ \!\left[ \left[ \mathcal{U}^\dag_{R_j}\!(H_j) , G \right] \!, \mathcal{U}^\dag_{R_i}\!(H_i) \right] \rho \bigg]\\
    & = \Tr\!\bigg[ \left[\left[ U^\dagger_{R_j} H_j U_{R_j},G\right] , U^\dagger_{R_i} H_i U_{R_j} \right] \rho\bigg].\label{eq:KM_app-alt}
\end{align}

We are now in a position to present an algorithm to estimate the second term of~\eqref{eq:KM_app} using its equivalent form shown in~\eqref{eq:KM_app-alt}. This term closely resembles the second term in the expression of \cite[Theorem~17]{minervini2025eqbm}, and thus the estimation algorithm is analogous to the one presented in \cite[Appendix I.3]{minervini2025eqbm}. So here we provide just a high-level description of the algorithm.  At its core, the algorithm relies on a quantum circuit that estimates the expected value of nested commutators of three operators, $\frac{1}{4} \left\langle \big[ \left[U_1 , H \right], U_0\big]\right\rangle_\rho$, where $H$ is Hermitian, and $U_0$ and $U_1$ are both Hermitian and unitary (refer to \cite[Appendix B]{minervini2025eqbm}). In this case, we choose $U_1 = U^\dagger_{R_j} H_j U_{R_j}$, $H = G$, and $U_0 = U^\dagger_{R_i} H_i U_{R_j}$. Accordingly, the quantum circuit that plays a role in estimating the quantity in~\eqref{eq:KM_app-alt} is depicted in Figure~\ref{fig:KM}. The algorithm involves running this circuit $N$ times, where $N$ is determined by the desired precision and error probability. The final estimation of the quantity~\eqref{eq:WY_app} is obtained by averaging the outputs of the $N$ runs and multiplying the result by $4 \left \| \alpha \right\|_1$, where for measuring $G$, we again adopt a sampling approach (see Remark~\ref{rem:measuring-G-phi}).

\end{document}